\documentclass[bj]{imsart}   
\usepackage{color,graphicx,verbatim,amsmath,amsthm,verbatim, amsfonts, url,cases,tikz,epsfig}
\usepackage[numbers,sort&compress]{natbib}

\startlocaldefs

\newcommand{\PP}{\mathbb{P}}
\newcommand{\sid}[1]{{\color{black} #1}}
\newcommand{\EE}{\mathbb{E}}
\newcommand{\RR}{\mathbb{R}}
\newcommand{\dd}{\mathrm{d}}
\newcommand{\ind}{\boldsymbol{1}}
\newcommand{\bZ}{\boldsymbol{Z}}

\newcommand{\bD}{\boldsymbol{D}}

\newcommand{\br}{\boldsymbol{r}}
\newcommand{\bv}{\boldsymbol{v}}
\newcommand{\bx}{\boldsymbol{x}}
\newcommand{\by}{\boldsymbol{y}}
\newcommand{\bu}{\boldsymbol{u}}

\newcommand{\bs}{\boldsymbol{s}}

\newcommand{\btheta}{\boldsymbol{\theta}}

\newcommand{\bolde}{\boldsymbol{e}}
\newcommand{\bxi}{\boldsymbol{\xi}}
\newcommand{\origin}{\boldsymbol{0}}
\newcommand{\bzero}{\boldsymbol{0}}
\numberwithin{equation}{section}
\numberwithin{table}{section}
\numberwithin{figure}{section}

\newcommand{\Din}{D^\text{in}}
\newcommand{\Dout}{D^\text{out}}

\newcommand{\convas}{\stackrel{\text{a.s.}}{\longrightarrow}}
\newcommand{\convp}{\stackrel{p}{\longrightarrow}}

\newtheorem{Theorem}{Theorem}[section]

\newtheorem{Definition}{Definition}[section]

\endlocaldefs

\allowdisplaybreaks

\begin{document}           
\bibliographystyle{imsart-number}

\begin{frontmatter}
\title{Preferential Attachment with Reciprocity: Properties and Estimation}
\runtitle{PA with Reciprocity}

\begin{aug}
\author[A]{\inits{D.}\fnms{Daniel} \snm{Cirkovic}\ead[label=e1,mark]{cirkovd@stat.tamu.edu}},
\author[A]{\inits{T.}\fnms{Tiandong} \snm{Wang}\ead[label=e2,mark]{twang@stat.tamu.edu}}
\and
\author[B]{\inits{S. I.}\fnms{Sidney} \snm{Resnick}\ead[label=e3,mark]{sir1@cornell.edu}}
\address[A]{Department of 
	Statistics, Texas 
	A\&M University, College 
		Station, TX 77843, U.S.. \printead{e1,e2}}

\address[B]{School of Operations Research and Information Engineering,
	 Cornell University, Ithaca, NY 14853, U.S.. \printead{e3}}
\end{aug}

\begin{abstract}
Reciprocity in social networks helps understand information exchange between two individuals, and 
indicates interaction patterns between pairs of users.
A recent study \citep{wang:resnick:2021a} indicates the reciprocity coefficient of a classical directed preferential attachment (PA) model does not match empirical evidence.
 In this paper, we extend the classical 3-scenario directed PA model by adding an additional parameter that controls the probability of creating a reciprocal edge. Unlike the model in \cite{wang:resnick:2021b}, our proposed model also allows edge creation between two existing nodes, making it a more realistic choice for fitting to real datasets. In addition to analysis of the theoretical properties of this PA model with reciprocity, we provide and compare two estimation procedures for the fitting of the extended model to both simulated and real datasets. 
The fitted models provide a good match with the empirical tail distributions of both in- and out-degrees. Other mismatched diagnostics suggest that further generalization of
the model is warranted.
\end{abstract}


\end{frontmatter}

\section{Introduction}

Reciprocal edges in a directed network correspond to mutual links between a pair of nodes, and
they describe the information exchange among users on social networks (cf. \cite{kilduff2006paradigm, molm2007building}). 
Consider, for instance, the interaction of users through wall posts on Facebook: when a user receives a message on his/her Facebook wall from a friend, then he/she is likely to reply to the post, thus forming a pair of mutual directed links between the two users. 
One classic quantitative measure of reciprocity is \emph{the reciprocity coefficient}, which is defined as
the proportion of reciprocated edges out of the total number of edges in a given network (cf. \cite{newman:etal:2001,jiangetal:2015, wasserman:faust:1994}), i.e.
for a directed graph $G=(V,E)$ with node set $V$ and edge set $E$, 
\begin{align}
\label{eq:def_recip}
r(G):= \frac{\bigl|\{(w,v)\in E: (v,w)\in E\}\bigr|}{|E|},
\end{align}
is the reciprocity coefficient.
Several empirical studies have used $r(G)$ as a summary statistic to
describe the reciprocity level for different types of networks, such
as the world trade web \cite{WorldTradeWeb:2003, WorldTradeWeb:2002},
neural networks \cite{White1986TheSO}, email networks
\cite{newman:etal:2001, Email:2002}, and social networks
\cite{Jeong:2000}. \sid{Examples of values for the reciprocity
  coefficient for different networks are given at
  \url{http://konect.cc/statistics/} and show a wide range of values in the
interval $[0,1]$; actual values depend on the dataset and the type of
network being sampled.}

In this paper, we consider social networks as our motivating examples.
The study in \cite{jiangetal:2015} compares eight types of networks, and concludes that online social networks, e.g. \cite{facebook:2009, flickr:2009, googleplus:2012}, tend to have a higher proportion of reciprocal edges than biological networks, communication networks, software call graphs and peer-to-peer networks.
Therefore, to model the dynamics of a social network, it is important to 
take into account the feature of having a high reciprocity coefficient.

When modeling directed social networks, the preferential attachment (PA) model (cf. \cite{bollobas:borgs:chayes:riordan:2003, krapivsky:redner:2001}) is an appealing choice in that it captures the scale-free property for complex networks, where both in- and out-degree distributions have Pareto-like tails
(cf. \cite{resnick:samorodnitsky:towsley:davis:willis:wan:2016, resnick:samorodnitsky:2015, wan:wang:davis:resnick:2017,wang:resnick:2019}). 
However, the analysis in \cite{wang:resnick:2021a} shows that the asymptotic behavior of the reciprocity coefficient in a classical directed PA model for
certain choices of the model parameters is close to 0.
This indicates a lack of fit for the classical PA model when the given network has a large proportion of reciprocal edges; for instance, the reciprocity coefficient for the network in \cite{facebook:2009} is 0.615. 
With such discrepancy in mind, a simple PA model with reciprocity
is proposed in \cite{wang:resnick:2021b}, which, however, does not allow edges between two existing nodes, but requires adding a new node at each step of the network evolution. This is an unrealistic assumption, making the model in \cite{wang:resnick:2021b} hardly applicable to real datasets. 
Therefore, in this paper, we consider a more realistic version of the
reciprocal PA model by taking into account edge creations between two
existing nodes. We also discuss the fitting of this model 
\sid{to both simulated and real network datasets.}

Our proposed model extends the classical 3-scenario directed PA model (cf. \cite{bollobas:borgs:chayes:riordan:2003,wan:wang:davis:resnick:2017}) in a way that whenever a new directed edge $(v,w)$ is formed following the PA rule, we add its reciprocal counterpart $(w,v)$ simultaneously with probability $\rho\in (0,1)$. However, the allowance of edge creation between existing nodes makes the derivation of theoretical properties more challenging, compared with the one given in \cite{wang:resnick:2021b}. Here we generalize methods in \cite{wang:resnick:2021b} by carefully embedding the in- and out-degree sequences into a sequence of two-type branching processes with immigration, and analyze the asymptotic dependence structure between large in- and out-degrees using embedding results. 

Furthermore, we provide two estimation approaches to fit the proposed model, namely likelihood based and extreme-value based methods. There are two major concerns when it comes to the likelihood based approach. 
First, the model suggests instantaneous creation of reciprocal edges and this is unlikely to be reflected in the data,
since, for example, there may have been already a number of wall posts among Facebook users during the time one post is sent and replied to.
Second, timestamp information does not
distinguish between an edge created by reciprocity or by the standard preferential attachment rule. Because of these two concerns, 
we combine the maximum likelihood method with a window estimator to produce parameter estimates.
Our second estimation method uses the derived asymptotic properties of large in- and out-degrees coupled with extreme value methods. 
Based on prior studies in \cite{wan:wang:davis:resnick:2017b}, \sid{we
expect} the extreme-value estimation method to be more robust against
model \sid{error} compared to likelihood-based approaches;
\sid{furthermore, it}
is  
applicable even when timetamps are coarse.
Both methods are also applied to a real dataset, the Facebook wall post data \citep{facebook:2009} available at
\url{http://konect.cc/networks/facebook-wosn-wall/}, and our fitted models provide reasonable fits for the marginal
in- and out-degree distributions.

The rest of the paper is organized as follows. We start with a description of the proposed model in Section~\ref{sec:model}, and Section~\ref{sec:embed} gives details on the multi-type branching process, which lays the foundation for the derivation of theoretical results stated in Section~\ref{sec:conv}.
We also present two estimation procedures in Section~\ref{sec:est}, both of which are then applied to simulated and real network datasets in Sections~\ref{sec:sim} and \ref{sec:real}, respectively. We provide important concluding remarks based on the fitting results in Section~\ref{sec:rmk}. 
Details on technical proofs are collected in the appendix.

\subsection{The Proposed PA Model with Reciprocity}\label{sec:model}
Initialize the model with graph $G(0)$, which consists of one node (labeled as Node 1) and a self-loop.
Let $G(n)$ denote the graph  after $n$ steps and
 $V(n)$ be the set of nodes in $G(n)$ with $V(0) = \{1\}$ and $|V(0)| = {1}$.
Denote the set of directed edges in $G(n)$ by $E(n)$ such that
an ordered pair $(w_1,w_2)\in E(n)$, $w_1,w_2\in V(n)$, 
represents a directed edge $w_1\mapsto w_2$. 
When $n=0$, we have $E(0) = \{(1,1)\}$.

Set $\bigl(\Din_w(n), \Dout_w(n)\bigr)$ to
be the in- and out-degrees of node $w\in V(n)$. 
We use the convention that $\Din_w(n) = \Dout_w(n) = 0$ if $w\notin V(n)$.
From $G(n)$ to $G(n+1)$, one of the following scenarios happens:
\begin{enumerate}
\item[(i)] With probability $\alpha$, 
we add a new node $|V(n)|+1$ with a directed edge $(|V(n)|+1,w)$, where $w\in V(n)$ is chosen with probability
\begin{equation}\label{eq:Din}
\frac{\Din_w(n)+\delta}{\sum_{w\in V(n)} (\Din_w(n)+\delta)}
= \frac{\Din_w(n)+\delta}{|E(n)|+\delta |V(n)|},
\end{equation}
and update the node set $V(n+1)=V(n)\cup\{|V(n)|+1\}$.
If, with probability $\rho\in (0,1)$, a reciprocal edge $(w, |V(n)|+1)$ is added, we
 update the edge set as
$E(n+1) = E(n)\cup \{(|V(n)|+1,w), (w,|V(n)|+1)\}$.
If the reciprocal edge is not created, set $E(n+1) = E(n)\cup \{(|V(n)|+1,w)\}$.
\item[(ii)] With probability $\gamma$, 
we add a new node $|V(n)|+1$ with a directed edge $(w,|V(n)|+1)$, where $w\in V(n)$ is chosen with probability
\begin{equation}\label{eq:Dout}
\frac{\Dout_w(n)+\delta}{\sum_{w\in V(n)} (\Dout_w(n)+\delta)}
= \frac{\Dout_w(n)+\delta}{|E(n)|+\delta |V(n)|},
\end{equation}
and update the node set $V(n+1)=V(n)\cup\{|V(n)|+1\}$.
If, with probability $\rho\in (0,1)$, a reciprocal edge $(|V(n)|+1,w)$ is added, we
update the edge set as 
$E(n+1) = E(n)\cup \{(|V(n)|+1,w), (w,|V(n)|+1)\}$.
If the reciprocal edge is not created, set $E(n+1) = E(n)\cup \{(w,|V(n)|+1)\}$.
\item[(iii)] With probability $\beta\in (0,1)$, we add an edge $(v,w)$ between two existing nodes $v,w\in V(n)$, 
with probability
\[
\frac{\Din_{w}(n)+\delta}{\sum_{w\in V(n)} (\Din_{w}(n)+\delta)}\frac{\Dout_{v}(n)+\delta}{\sum_{v\in V(n)} (\Dout_{v}(n)+\delta)}
= \frac{\Din_{w}(n)+\delta}{|E(n)|+\delta |V(n)|}\frac{\Dout_{v}(n)+\delta}{|E(n)|+\delta |V(n)|}. 
\]
Then with probability $\rho\in (0,1)$, we add a reciprocal edge $(w,v)$. We then update 
the edge set as
$E(n+1) = E(n)\cup \{(v,w), (w,v)\}$.
If the reciprocal edge is not created, then $E(n+1) = E(n)\cup \{(v,w)\}$.
\end{enumerate}
We assume the offset parameter $\delta$ takes the same value for both in- and out-degrees, and
we do not assign different step indices when the reciprocal edge is also added from $G(n)$ to $G(n+1)$.

By the definition of reciprocity coefficient in \eqref{eq:def_recip}, we see that a.s.
\[
\lim_{n\to\infty} r(G(n))\ge \frac{2\rho}{1+\rho}.
\]
Therefore, by introducing a reciprocal component, we obtain a lower bound for the reciprocity coefficient, overcoming the drawback in the classical directed PA model where the reciprocity coefficient may be close to 0 for certain choices of parameters.

\section{Markov Branching with Immigration}\label{sec:embed}
\sid{For a PA model with reciprocity, the derivation of asymptotic results in
Section~\ref{sec:conv} depends on embedding the in- and out-degree
sequences into a family of independent multi-type \sid{Markov} 
branching processes with immigration (MBI).
The required embedding framework is more elaborate
than the one in \cite{wang:resnick:2021b} which lacked  the $\beta$-scenario,
and thus details 
are deferred to Appendix \ref{append:embed}. However,
asymptotic results on degree counts have limits expressed in terms of
an MBI process, so in preparation for Section \ref{sec:conv}, 
we give a brief
description of the MBI process.}

\subsection{MBI Ingredients}\label{subsec:MBI}
The two-type MBI process
$\{\bxi_{\delta}(t) = (\xi^{(1)}_{\delta}(t),\xi^{(2)}_{\delta}(t)):
t\ge0\}$ is a Markov branching process.  The general setup
of continuous-time multitype branching processes without immigration
is reviewed in \cite[Chapter V]{athreya:ney:1972} and discussions on
the MBI process are included, for instance, in
\cite{rabehasaina:2021, wang:resnick:2021b}. The process
$\bxi_{\delta} (\cdot)$ is designed to mimic evolution of  in- and
out-degrees of a fixed node.
 Life time parameters of 
$\xi^{(1)}_{\delta}(\cdot)$ and $\xi^{(2)}_{\delta}(\cdot)$ are
$a_1=\alpha+\beta, a_2=\beta+\gamma$, respectively and the branching
structure is given by the offspring generating functions
\begin{align}
f^{(1)}(\bs) &= (1-\rho)s_1^2 + \rho s_1^2s_2 ,\label{eq:pgf1}\\
f^{(2)}(\bs) &= (1-\rho)s_2^2 + \rho s_1 s_2^2 ,\label{eq:pgf2}
\end{align}
for $\bs=(s_1,s_2)\in [0,1]^2$.
Equation~\eqref{eq:pgf1} gives that at the end of the life time of a type 1 particle, with 
probability $1-\rho$, it will split into two type 1 particles, increasing the total number of
type 1 particles by 1.
With probability $\rho$, a type 1 particle will give birth to 2 type 1 particles and 1 type 2 particle upon its death, which increases the total numbers of type 1 and 2 particles both by 1. 
Similar interpretations also apply to \eqref{eq:pgf2}.
Immigration events occur at Poisson rate
$(1+\beta)\delta$ which gives the immigration parameter of $\bxi_\delta
(\cdot)$. When an immigration event occurs, $\bxi_\delta(\cdot)$
changes by $(1,0), (0,1)$ or $(1,1) $ according to the distribution
\begin{align}
\label{eq:p0r}
p_0(\br) \equiv \left(\frac{(\alpha+\beta)(1-\rho)}{1+\beta}\right)^{\ind_{\{\br = (1,0)\}}}\left(\frac{(\beta+\gamma)(1-\rho)}{1+\beta}\right)^{\ind_{\{\br = (0,1)\}}}\rho^{\ind_{\{\br = (1,1)\}}},
\end{align}
and the branching structure for immigrants is the same as in
\eqref{eq:pgf1} and \eqref{eq:pgf2}. 
The initial values of $\{\bxi_{\delta}(\cdot)$ are specified 
in \eqref{eq:pgfxi0}.

Conditioning on the current state $\bx\equiv (x_1,x_2)$,
the jump probability of $\bxi_{\delta}(\cdot)$ from $\bx$ to
$\bx+(1,0)$ is a result of exponential competitions between the death
of a type 1 particle and the arrival of a new immigration event
$(1,0)$. 
Therefore, we have 
\begin{align}
P&(\bx, \bx+(1,0)) = \PP(\text{One type 1 particle dies first, giving birth to 2 type 1 particles})\nonumber\\
&\quad+\PP(\text{One immigration event $(1,0)$ arrives first})\nonumber\\
&=\frac{(1-\rho)(\alpha+\beta)x_1}{(\alpha+\beta) x_1+ (\beta+\gamma) x_2+(1+\beta)\delta}
+\frac{(\alpha+\beta)(1-\rho)}{1+\beta} \frac{(1+\beta)\delta}{(\alpha+\beta) x_1+ (\beta+\gamma) x_2+(1+\beta)\delta}\nonumber\\
&= \frac{(\alpha+\beta)(1-\rho)(x_1+\delta)}{(\alpha+\beta) x_1+ (\beta+\gamma) x_2+(1+\beta)\delta}.
\label{eq:jump1}
\end{align}
Following a similar reasoning, we see that
\begin{align}
&P(\bx, \bx+(0,1)) = \PP(\text{One type 2 particle dies first, giving birth to 2 type 2 particles})\nonumber\\
&\quad+\PP(\text{One immigration event $(0,1)$ arrives first})\nonumber\\
&=\frac{(1-\rho)(\beta+\gamma)x_1}{(\alpha+\beta) x_1+ (\beta+\gamma) x_2+(1+\beta)\delta}
+\frac{(\beta+\gamma)(1-\rho)}{1+\beta} \frac{(1+\beta)\delta}{(\alpha+\beta) x_1+ (\beta+\gamma) x_2+(1+\beta)\delta}\nonumber\\
&= \frac{(\beta+\gamma)(1-\rho)(x_1+\delta)}{(\alpha+\beta) x_1+ (\beta+\gamma) x_2+(1+\beta)\delta},
\label{eq:jump2}
\end{align}
and 
\begin{align}
&P(\bx, \bx+(1,1)) \nonumber\\
&= \PP(\text{One type 1 particle dies first, giving birth to 2 type 1 particles and 1 type 2 particle})\nonumber\\
&\,+ \PP(\text{One type 2 particle dies first, giving birth to 2 type 2 particles and 1 type 1 particle})\nonumber\\
&\,+\PP(\text{One immigration event $(1,1)$ arrives first})\nonumber\\
&=\frac{\rho}{(\alpha+\beta) x_1+ (\beta+\gamma) x_2+(1+\beta)\delta}
\left((\alpha+\beta) x_1+ (\beta+\gamma) x_2+(1+\beta)\delta\right)
=\rho.\label{eq:jump3}
\end{align}

Using the generating functions in \eqref{eq:pgf1} and \eqref{eq:pgf2}, we follow
\cite[Chapter V.7.2]{athreya:ney:1972} to define a matrix
\begin{equation}\label{eq:defA}
A= 
\begin{bmatrix}
\alpha+\beta & (\alpha+\beta)\rho\\
(\beta+\gamma)\rho & \beta+\gamma
\end{bmatrix},
\end{equation}
which specifies the branching structure of the MBI process $\bxi_{\delta}(\cdot)$.
By the Perron--Frobenius theorem (cf. \cite[Theorem V.2.1]{athreya:ney:1972}),
the matrix $A$ has a largest positive eigenvalue with multiplicity 1, which is
\begin{align}\label{eq:lambda1}
\lambda_1 = \frac{1}{2}\left(1+\beta + \sqrt{(\alpha-\gamma)^2+4(\alpha+\beta)(\beta+\gamma)\rho^2}\right) =: \frac{1}{2}\left(1+\beta + \sqrt{D_0}\right).
\end{align}
Let $\bv,\bu$ be the left and right eigenvectors of $\lambda_1$ respectively, with all coordinates strictly positive, and $\bu^T\ind =1$, $\bu^T\bv =1$.
Applying \cite[Theorem~1]{wang:resnick:2021b} gives that
there exists some finite positive random variable $L$ such that
\begin{align}\label{eq:conv_MBI}
e^{-\lambda_1 t}\bxi_{\delta}(t)\convas L \bv, 
\end{align}
where
\begin{align}
\bv &= \left(\frac{\alpha+2(\beta+\gamma)\rho-\gamma+\sqrt{D_0}}{2\sqrt{D_0}}\right)\begin{bmatrix}
1\medskip\\
\frac{\gamma-\alpha+\sqrt{D_0}}{2\rho(\beta+\gamma)}
\end{bmatrix}.\label{eq:defv}
\end{align}
When $\lambda_1\ge\log 2$, applying Proposition 1 in \cite{wang:resnick:2021b}, we also have
that for an integer $q\ge 1$,
\begin{align}\label{eq:MBI_moment}
\sup_{t\ge 0}e^{-\lambda_1 qt} \EE\left[\left({\xi}^{(i)}_{\delta}(t)\right)^q\right]
<\infty,\qquad  i=1,2.
\end{align}
Later in Theorem~\ref{thm:MRV}, Equation~\eqref{eq:MBI_moment} is an important condition to prove the asymptotic behavior of limiting in- and out-degrees in the reciprocal PA model.

\section{Convergence of Degree Counts and Multivariate Regular Variation}
\label{sec:conv}

\sid{Based on the description of the MBI process, 
we provide the convergence of the joint in- and out-degree counts}
in a reciprocal PA model.
We then study the asymptotic
dependence structure for large in- and out-degrees.

\subsection{Convergence of Degree Counts}
In the current section, we focus on the limiting behavior of joint degree counts:
\begin{align*}
N_{m,l}(n) = \sum_{w=1}^{|V(n)|} \ind_{\left\{(\Din_w(n),\Dout_w(n))=(m,l)\right\}}.
\end{align*}
Using the embedding result, the following theorem shows the convergence
of $N_{m,l}(n)/n$.
\begin{Theorem}\label{thm:limitNij}
Suppose that $\{\widetilde{\bxi}_\delta(t):t\ge 0\}$ is a two-type MBI
process with the
\sid{branching mechanism described in Section \ref{subsec:MBI} and
initialization described in Equation~\eqref{eq:pgfxi0}.}
Then as $n\to\infty$, we have
for $m,l\ge 0$,
\begin{align}\label{eq:Nij}
\frac{N_{m,l}(n)}{n}&\convp (1-\beta)\int_0^\infty (1+\rho+(1-\beta)\delta)e^{-t(1+\rho+(1-\beta)\delta)}
\PP\left(\widetilde{\bxi}_\delta(t)=(m,l)\right)\dd t \nonumber \\
&= : (1-\beta)\PP\left(\big(\mathcal{I},\mathcal{O}\bigr) = (m,l)\right).
\end{align}
\end{Theorem}
Since $|V(n)|/n\convas 1-\beta$ as $n\to\infty$, then it follows from \eqref{eq:Nij} that 
\begin{align}\label{eq:Nij_scale}
\frac{N_{m,l}(n)}{|V(n)|}\convp\PP\left(\big(\mathcal{I},\mathcal{O}\bigr) = (m,l)\right)
= \PP\left(\widetilde{\bxi}_\delta(E)=(m,l)\right),
\end{align}
where 
$E$ is an exponential random variable with rate $1+\rho+(1-\beta)\delta$, independent from $\widetilde{\bxi}_\delta(\cdot)$.

The proof of Theorem~\ref{thm:limitNij} is given in
Appendix~\ref{proof2} \sid{and is based on embedding the in- and
  out-degree sequences in a family of MBI processes and using
  knowledge of asymptotics of MBI processes.
}



\subsection{Multivariate Regular Variation of $(\mathcal{I},\mathcal{O})$}
Based on Theorem~\ref{thm:limitNij}, we now show that the limiting in-
and out-degree counts, $(\mathcal{I},\mathcal{O})$ in \eqref{eq:Nij},
are jointly heavy tailed. To formalize our analysis, we provide some
useful definitions related to \emph{multivariate regular variation}
(MRV) \sid{of measures.}

Suppose that $\mathbb{C}_0\subset\mathbb{C}\subset\mathbb{R}_+^2$ are two closed cones, and we first give the definition of $\mathbb{M}$-convergence in Definition~\ref{def:Mconv}
(cf. \cite{lindskog:resnick:roy:2014,hult:lindskog:2006a,das:mitra:resnick:2013,kulik:soulier:2020,basrak:planinic:2019}) on $\mathbb{C}\setminus \mathbb{C}_0$, which lays the theoretical foundation of regularly varying measures (cf. Definition~\ref{def:MRV}).

\begin{Definition}\label{def:Mconv}
Let $\mathbb{M}(\mathbb{C}\setminus \mathbb{C}_0)$ be the set of Borel
measures on $\mathbb{C}\setminus \mathbb{C}_0$ which are finite on
sets bounded away from $\mathbb{C}_0$, and
$\mathcal{C}(\mathbb{C}\setminus \mathbb{C}_0)$ be the set of
continuous, bounded, non-negative functions on $\mathbb{C}\setminus
\mathbb{C}_0$ whose supports are bounded away from  $\mathbb{C}_0$. Then for $\mu_n,\mu \in \mathbb{M}(\mathbb{C}\setminus
\mathbb{C}_0)$, we say $\mu_n \to \mu$ in
$\mathbb{M}(\mathbb{C}\setminus \mathbb{C}_0)$, if $\int f\dd
\mu_n\to\int f\dd \mu$ for all $f\in \mathcal{C}(\mathbb{C}\setminus
\mathbb{C}_0)$. 
\end{Definition} 
 
\begin{Definition}\label{def:MRV}
The distribution of a
 random vector $\bZ$ on $\mathbb{R}_+^2$, i.e.
$\PP(\bZ\in\cdot)$, 
  is (standard) regularly varying on $\mathbb{C}\setminus \mathbb{C}_0$ with index $c>0$ (written as $\PP(\bZ\in\cdot)\in \text{MRV}(c, b(t), \nu, \mathbb{C}\setminus \mathbb{C}_0)$) if there exists some scaling function $b(t)\in \text{RV}_{1/c}$ and a limit measure $\nu(\cdot)\in \mathbb{M}(\mathbb{C}\setminus \mathbb{C}_0)$ such that
as $t\to\infty$,
\begin{equation}\label{eq:def_mrv}
t\PP(\bZ/b(t)\in\cdot)\rightarrow \nu(\cdot),\qquad\text{in }\mathbb{M}(\mathbb{C}\setminus \mathbb{C}_0).
\end{equation}
\end{Definition}

When analyzing the asymptotic dependence between components of a
bivariate random vector
$\bZ$ satisfying \eqref{eq:def_mrv},  it is often informative  
to make a polar coordinate transform and
consider the transformed points located on the $L_1$ unit sphere
\begin{align}
\label{eq:map_L1}
(x,y)\mapsto\left(\frac{x}{x+y},\frac{y}{x+y}\right),
\end{align}
after thresholding the data according to the
$L_1$ norm.
The plot of the transformed points is referred to as the diamond plot,
which provides a visualization of dependence. 
Also, provided that $x+y$ is larger than some predetermined threshold, the density plot of thresholded values $x/(x+y)$ is called the angular density plot.
These plots characterize the asymptotic dependence structure for extremal observations.

The next theorem states that the limiting pair
$(\mathcal{I},\mathcal{O})$ \sid{has a distribution that} is jointly regularly varying, and 
applying the transformation in \eqref{eq:map_L1} to nodes with large
in- and out-degrees, we find that the angular density plot concentrates around 
some particular value. In the terminology of \cite{das:resnick:2017}, this indicates that the limiting in- and out-degree pair has \emph{full asymptotic dependence}.

\begin{Theorem}\label{thm:MRV}
Let $(\mathcal{I},\mathcal{O})$ be as in \eqref{eq:Nij} and $\lambda_1$ as in \eqref{eq:lambda1}.
If $\lambda_1\ge \log2$, then
\begin{align}
\label{eq:MRV_IO}
\PP\bigl((\mathcal{I},\mathcal{O})\in\cdot\bigr) \in \text{MRV}\left(\frac{1+\rho+\delta (1-\beta)}{\lambda_1}, t^{\lambda_1/(1+\rho+\delta(1-\beta))}, \mu, \mathbb{R}_+^2\setminus\{\origin\}\right),
\end{align}
where the 
limit measure $\mu \in \mathbb{M} (\RR_+^2\setminus \{\origin\})$
satisfies for any $f\in \mathcal{C} (\RR_+^2\setminus \{\origin\})$,
\begin{align}\label{eq:limit_mu}
  \mu(f) =\int_0^\infty \EE \bigl( f(y\widetilde L
  \bv) \bigr)\nu_{(1+\rho+\delta(1-\beta))/\lambda_1} (\dd y),
\end{align}
\sid{and}
$\widetilde{L}$ satisfies
$e^{-\lambda_1t}\widetilde{\bxi}_\delta(t)\convas \widetilde{L}\bv$.
Also,
 $(v_1,v_2)^T\equiv \bv$ is given in \eqref{eq:defv}. Since
$\widetilde L$ is one-dimensional and $\bv$ is deterministic, the
distribution of $\widetilde L \bv$ concentrates on a one-dimensional
subspace and therefore
$\mu(\cdot)$ concentrates Pareto mass on the line $y=bx$ where
\begin{equation}\label{e:defa}
  b=\frac{v_2y\widetilde L}{v_1y\widetilde L} =\frac{v_2}{v_1} =
   \frac{  \gamma-\alpha+\sqrt D_0}{2(\beta+\gamma)\rho},
\end{equation}
with $D_0$ as defined in \eqref{eq:lambda1}.
In addition, there exists some constant $C>0$ such that
$\mu\bigl((x,\infty)\times [0,\infty)\bigr) = C
x^{-(1+\rho+\delta(1-\beta))/\lambda_1}$, $x>0$.

Switching to $L_1$-polar coordinates via the transformation
$$V:(x,y)\mapsto \Bigl(\frac{(x,y)}{x+y}, (x+y) \Bigr) =: (\btheta, r)$$
from $\RR_+^2 \setminus \{\bzero\} \mapsto  \{(x,y) \in
\RR_+^2\setminus \{\bzero\} : x+y=1\} \times
(0,\infty)=:\aleph_0\times (0,\infty)$, we find with
\begin{align}
\label{eq:angle}
\btheta_0=\Bigl( \frac{v_1}{v_1+v_2},\frac{v_2}{v_1+v_2} \Bigr)
\end{align}
that
$$\mu\circ V^{-1} (\dd\btheta, \dd r)=
\epsilon_{\btheta_0} (d\btheta) \widetilde{C} \nu_{(1+\rho+\delta(1-\beta))/\lambda_1}
(\dd r)$$
where $\epsilon_{\btheta_0}(\cdot)$ is the Dirac probabilty measure
concentrating all mass on $\btheta_0$ and
\[
\widetilde{C}= \int_0^\infty \frac{(1+\rho+\delta(1-\beta))}{\lambda_1}z^{-1-(1+\rho+\delta(1-\beta))/\lambda_1}
\times\PP\left(\widetilde{L}>\frac{1}{z(v_1+v_2)}\right)\dd z.
\]
\end{Theorem}
\begin{proof}
We will prove Theorem~\ref{thm:MRV} by
 applying Theorem~\ref{th:extendBrei}, so we first need to check $\PP(\widetilde{L}>0)=1$.
Based on Theorem~V.7.2 and Equation~(V.25) in \cite{athreya:ney:1972}, we apply a similar reasoning as in the proof of \cite[Theorem 3]{wang:resnick:2021b} to conclude that $\widetilde{L}>0$ a.s..
 
From \eqref{eq:Nij}, we see that
$
\big(\mathcal{I},\mathcal{O}\bigr) 
\stackrel{d}{=} \widetilde{\bxi}_\delta\left(\widetilde{T}\right)$,
where $\widetilde{T}$ is an exponential random variable with rate
$1+\rho+\delta(1-\beta)$, independent from the $\widetilde{\bxi}_\delta(\cdot)$
process.  The proof of 
 \eqref{eq:MRV_IO} and \eqref{eq:limit_mu} is an application of Theorem
 \ref{th:extendBrei} after making the identifications
 \begin{align*}
&\bxi (t)=t^{-1} \widetilde{\bxi}_\delta\left(\frac{1}{\lambda_1} \log t\right), 
                                                                      &&\bxi_\infty =\widetilde{L}\bv,
   &&X=e^{\lambda_1 \widetilde{T}},\\
   &b(t)=t^{\lambda_1/(1+\rho+\delta(1-\beta))},&&
c=(1+\rho+\delta(1-\beta))/\lambda_1 && {}.
 \end{align*}
 The remaining piece is to show \eqref{e:extraCond} in this
 context. In fact, by \eqref{eq:MBI_moment}, we see that for
any $\delta \ge 0$ and any $q=1,2,\ldots$, there exists some constant $K(\delta,q)>0$ such that 
\begin{align*}
\sup_{t\ge 0}e^{-\lambda_1 qt} \EE\left[\left(\widetilde{\xi}^{(1)}_\delta(t)\right)^q\right]
\le K(\delta, q),
\end{align*}
thus verifying the condition in \eqref{e:extraCond}.

The comments about where $\mu(\cdot)$
concentrates and the representation of $\mu $ in polar coordinates is
standard; see, for example, \cite[p. 292]{lindskog:resnick:roy:2014}
and this 
completes the proof of Theorem~\ref{thm:MRV}.
\end{proof}

The MRV results in Theorem~\ref{thm:MRV} provide the theoretical foundation for the proposed estimation approaches in Section~\ref{sec:est}, and we will apply these approaches to both simulated and real datasets in Sections \ref{sec:sim} and \ref{sec:real}, respectively.

\section{Estimation}\label{sec:est}
In this section, we discuss estimation methods for the PA model with reciprocity.
Suppose we start with an arbitrary initial graph $G(n_0) = (V(n_0),E(n_0))$ with $|V(n_0)|\ge 1$ nodes and $|E(n_0)|\ge 1$ edges. Let the graph
 evolve according to the PA with reciprocity rule, denoting $G(n)$ to be the graph at step $n \geq n_0$. Define $e_n := E(n)\setminus E(n - 1)$ as the newly added edge(s) as the graph evolves from $G(n-1)$ to $G(n)$. Note that either $e_{n} = (s_n, t_n)$ or $e_{n} = \{ (s_n, t_n), (t_n, s_n) \}$, depending on whether or not a reciprocal edge is created  at step $n$. 
For $e_n = \{(s_n, t_n), (t_n, s_n)\}$, we 
denote $(s_n, t_n)$ as the parent edge.

Suppose that we observe $G(k)$, $k=n_0,\dots, n$, and that each edge is accompanied by a timestamp. 
For $k = n_0 + 1,\ldots, n$, and a non-reciprocal edge $e_k$, let $J_k=1,2,3$ specify whether $e_k$ is created under $\alpha$-, $\beta$-, and $\gamma$-scenarios, respectively. 
For a pair of reciprocal edges $e_k = \{(s_k, t_k), (t_k, s_k)\}$,
use $J_k=1,2,3$ to describe the three edge creation scenarios of the
parent edge correspondingly.  
Adopting notations from \sid{the construction in Appendix~\ref{append:embed}},
we use Bernoulli random variables $\{R_k: k = n_0 + 1, \ldots, n\}$ to denote whether a reciprocal edge is created at step $k$. 
The set $\mathcal{R} = \{k : R_k = 1 \}$ collects all steps at which a reciprocal edge is generated by the proposed model. 
Then for each $k\in \mathcal{R}$, two edges, $(s_k, t_k)$ and $(t_k, s_k)$, are generated at the same step in the reciprocal PA model. 
The likelihood function becomes:
\begin{align}
\label{eq:likelihood}
& L(\alpha, \beta, \rho, \delta ; G(n_0), (e_k)_{k = n_0 + 1}^n)\nonumber \\
& = \prod_{k = n_0 + 1}^n \left( \alpha \frac{\Din_{t_k}(k-1) + \delta}{|E(k - 1)| + \delta |V(k - 1)|} \right)^{\mathbf{1}_{\{ J_k = 1 \}}} \nonumber\\
& \ \ \ \ \ \times \prod_{k = n_0 + 1}^n \left( \beta \left( \frac{\Din_{t_k}(k-1) + \delta}{|E(k - 1)| + \delta |V(k- 1)|} \right) \left( \frac{\Dout_{s_k}(k-1) + \delta}{|E(k - 1)| + \delta |V(k - 1)|} \right) \right)^{\mathbf{1}_{\{ J_k = 2 \}}}\nonumber \\
& \ \ \ \ \ \ \ \ \ \ \times \prod_{k = n_0 + 1}^n \left( ( 1 - \alpha - \beta) \frac{\Dout_{s_k}(k-1) + \delta}{|E(k- 1)| + \delta |V(k- 1)|} \right)^{\mathbf{1}_{\{ J_k = 3 \}}} \nonumber\\
& \ \ \ \ \ \ \ \ \ \ \ \ \ \ \ \ \times  \prod_{k = n_0 + 1}^n \rho^{\mathbf{1}_{\{ R_k = 1 \}}} (1 - \rho)^{\mathbf{1}_{\{ R_k = 0 \}}},
\end{align}
with log-likelihood 
\begin{align}
\label{loglikelihood}
& \ell(\alpha, \beta, \rho, \delta ; G(n_0), (e_k)_{k = n_0 + 1}^n) \nonumber\\
&= \log \alpha \sum_{k = n_0 + 1}^n \mathbf{1}_{\{J_k = 1 \}} + \log \beta \sum_{k = n_0 + 1}^n \mathbf{1}_{\{J_k = 2\}} + \log (1 - \alpha -\beta) \sum_{k = n_0 + 1}^n \mathbf{1}_{\{J_k = 3 \}}\nonumber \\
& \ \ \ \ \ + \log \rho \sum_{k = n_0 + 1}^n \mathbf{1}_{\{R_k = 1 \}} + \log ( 1 - \rho) \sum_{k = n_0 + 1}^n \mathbf{1}_{\{R_k = 0 \}}\nonumber  \\
& \ \ \ \ \ \ \ \ \ \ + \sum_{k = n_0 + 1}^n \log \left( \Din_{t_k}(k-1) + \delta \right)\mathbf{1}_{\{J_k \in \{ 1, 2 \} \}}\nonumber \\
& \ \ \ \ \ \ \ \ \ \ \ \ \ \ \ \ + \sum_{k = n_0 + 1}^n \log \left( \Dout_{s_k}(k-1) + \delta \right)\mathbf{1}_{\{J_k \in \{ 2, 3 \} \}} \nonumber\\
& \ \ \ \ \ \ \ \ \ \ \ \ \ \ \ \ \ \ \ \ \ - \sum_{k = n_0 + 1}^n \log \left( |E(k - 1)| + \delta |V(k - 1)| \right)(1 + \mathbf{1}_{\{J_k = 2 \}}).
\end{align}
The score equations for $\alpha$ and $\beta$ 
give the corresponding MLEs: 
\begin{align}\label{eq:MLE_init}
\hat{\alpha}^{MLE} = \frac{1}{n - n_0} \sum_{t = n_0 + 1}^n \mathbf{1}_{\{J_t = 1 \}}, 
\qquad
\hat{\beta}^{MLE} = \frac{1}{n - n_0} \sum_{t = n_0 + 1}^n \mathbf{1}_{\{J_t = 2\}}.
\end{align}
The portion of the likelihood in which $\rho$ contributes is a typical Bernoulli likelihood and can be maximized independently to obtain $\hat{\rho}^{MLE} = \frac{1}{n - n_0} \sum_{t = n_0 + 1}^n \mathbf{1}_{\{R_t = 1 \}}$. By the strong law of large numbers, $\hat{\alpha}^{MLE}, \hat{\beta}^{MLE}$ and  $\hat{\rho}^{MLE}$ are all strongly consistent for their respective targets. 
In addition, the score equation for $\delta$ is
\begin{equation}
\label{eq:deltascore}
\begin{split}
&\frac{\partial}{\partial \delta} \ell(\alpha, \beta, \rho, \delta \vert G(n_0), (e_k)_{k = n_0 + 1}^n) \\
&= \sum_{k = n_0 + 1}^n \frac{\mathbf{1}_{\{J_k \in \{ 1, 2 \} \}}}{\Din_{t_k}(k-1) + \delta} + \sum_{k = n_0 + 1}^n \frac{\mathbf{1}_{\{J_k \in \{ 2, 3 \} \}}}{\Dout_{s_k}(k-1) + \delta} - \sum_{k = n_0 + 1}^n \frac{(1 + \mathbf{1}_{\{J_k = 2 \}})|V(k - 1)|}{|E(k - 1)| + \delta |V(k - 1)|}.
\end{split}
\end{equation}
We then set \eqref{eq:deltascore} to zero and solve the equation numerically to obtain
$\hat{\delta}^{MLE}$. 
It is worthwhile noting that unlike the discussion in \cite{wan:wang:davis:resnick:2017}, the extra reciprocal component in the PA model makes theoretical analyses on the consistency and the asymptotic normality of $\hat{\delta}^{MLE}$ less tractable. 

In Section~\ref{subsec:est_real}, we couple the likelihood-based method with a window-based estimator, and also give an extreme-value based estimation approach in Section~\ref{subsec:ev_est}. Then we discuss properties of these estimators through a simulation study in Section~\ref{sec:sim}.

\subsection{Window Estimators}\label{subsec:est_real}

The proposed reciprocal PA model assumes an instant coin flip associated with each newly created edge to determine whether the reciprocal edge will be added. 
However, in large social networks such as Facebook and Twitter, during the time a message is sent and replied to between two users, there may have been multiple interactions among other users. 
Hence, when we observe the creation of a reciprocal edge based on the edge list obtained, it is difficult to know whether the edge is created due to reciprocity 
or the $\beta$-scenario. 

In what follows, we propose a window-based estimation approach so that whenever a reciprocated edge is observed within a window of step length $w$, we characterize it as a reciprocal event and move it back to the step at which its parent edge is added. After shifting back reciprocal edges within each window, we obtain a modified edge list which agrees with the model assumption required by the reciprocal PA model, thus making likelihood-based estimation methods plausible.

To give a detailed proposal of our window estimator, we first fix a window of length $w$.
If a parent edge in $e_k$ has a reciprocal counterpart in an edge set $e_{k + 1}, \dots, e_{k + w}$, we attribute the event $\{ R_k = 1 \}$ to $e_k$ and reallocate the first occurrence of a reciprocal edge to $e_k$. This is done for every $k = n_0 + 1, \dots, n$, and we set $\{R_k=0\}$ if $k+w>n$. Note that once an edge is labeled as reciprocated, neither it nor its parent can be used to denote another reciprocal edge. The rest of the non-reciprocated edges are then labeled via $J_k$ according to the PA evolution. 
Upon relabeling, we have a new edge list, $\{e^w_k\}$, which is in alignment with the structure required in Section \ref{sec:model} and we use the maximum likelihood estimates (MLE's) to obtain the parameter estimates, 
$\hat{\boldsymbol{\theta}}_w= (\hat{\alpha}_w, \hat{\beta}_w,\hat{\rho}_w,\hat{\delta}_w)$.
We provide one example in Table~\ref{tab:time_index}, where the observed edge matrix has 10 rows.
By choosing windows of fixed length $w=2$, 
we give the interpreted step indices in Table~\ref{tab:time_index} as well as the corresponding $\{e^w_k\}$.

\begin{table}
\centering
\begin{tabular}{c|ccc||cc}
\hline
Edges & Interpreted $k$ & Interpreted $J_k$ & Interpreted $R_k$ & $k$ & $e^w_k$\\
\hline
$(1,1)$ & 0 & - & - & 0 & $(1,1)$\\
$(1,1)$ & 1 & 2 & 0 & 1 & $(1,1)$\\
$(1,2)$ & 2 & 3 & 0 & 2 & $(1,2)$\\
$(3,2)$ & 3 & 1 & 1 & 3 & $(3,2)$\\
$(1,4)$ & 4 & 3 & 1 & 3 & $(2,3)$\\
$(2,3)$ & 3 & - & - & 4 & $(1,4)$\\
$(5,1)$ & 5 & 1 & 1 & 4 & $(4,1)$\\
$(4,1)$ & 4 & - & - & 5 & $(5,1)$\\
$(2,6)$ & 6 & 3 & 0 & 5 & $(1,5)$\\
$(1,5)$ & 5 & - & - & 6 & $(2,6)$\\
\hline
\end{tabular}
\caption{Given an edge matrix of 10 rows, 
we use a window of fixed length $w=2$ to obtain $\{e^w_k\}$, which is in alignment with the structure required in Section \ref{sec:model}.}
\label{tab:time_index}
\end{table}

Since the true value of $w$ is unknown for real datasets, we proceed
by first specifying a possible range, $\mathcal{W}$, for the window
length, then offering two \sid{selection} criteria to decide an optimal
$w$. The first proposed criterion is based on the likelihood
principle, where we choose the optimal window length, $w_L^\star$ (the
subscript $L$ stands for ``Likelihood''), as: 
\begin{align}
\label{eq:Wstar}
w_L^\star = \underset{w\in \mathcal{W}}{\mathrm{argmax}}\,\ell\left(\hat{\boldsymbol{\theta}}_w; G(n_0), (e^w_k)_{k = n_0 + 1}^n\right),
\end{align}
and $\{e^w_k\}$ denotes the interpreted edge list
using a window of length $w$ \sid{that conforms to the model.}

The second \sid{selection} criterion  choose\sid{s} the $w$ that best reflects the
theoretical MRV index $\iota$ as given in \eqref{eq:MRV_IO}. 
\sid{We start by estimating $\iota$ using 
the the minimum distance method given in  \cite{clauset:shalizi:newman:2009}
applied to the in- and out-degree sequences.}  \sid{Then} for a given $w$,
we compute the corresponding window estimates,
$\hat{\boldsymbol{\theta}}_w$, and  \sid{compute} $\iota$ by plugging
$\hat{\boldsymbol{\theta}}_w$ into the $\iota$ expression in
\eqref{eq:MRV_IO} \sid{thereby yielding} the estimate $\hat{\iota}_w$.
The optimal window length  minimizes the discrepancy between these two estimates, i.e. 
\begin{align}
\label{eq:WstarTail}
w_T^\star = \underset{w\in \mathcal{W}}{\mathrm{argmin}} \bigl\lvert \hat{\iota} - \hat{\iota}_w  \bigr\rvert,
\end{align}
where the subscript $T$ stands for ``Tail''.
Then $\hat{\boldsymbol{\theta}}_{w_L^\star}$ and
$\hat{\boldsymbol{\theta}}_{w_T^\star}$ are the two optimal window
estimates corresponding to the two \sid{selection} criteria in
\eqref{eq:Wstar} and \eqref{eq:WstarTail}, respectively. 
Later in Section~\ref{sec:sim}, we will examine the validity of these
two criteria through simulation studies.

We remark that \sid{for programming ease}  the window length $w$ in
our method is an integer and steps
refer to edge evolution.
\sid{However, our methods could be adapted so  $w\in \RR_+$ and the window refers to timestamps.}


\subsection{Extreme-Value Estimation Method}\label{subsec:ev_est}

While the window estimator is based on likelihood principles, most genuine datasets require modifications to be plausibly generated by the proposed PA model with reciprocity. 
If the timestamp of a given dataset is coarse (i.e. multiple edges are labeled with the same timestamp), then the window estimator becomes inapplicable since we cannot decide the exact order of edge creation. 
Therefore, it is useful to have a second method, and we propose an extreme-value based estimation approach, motivated by Theorem~\ref{thm:MRV}. This approach relies less heavily on the knowledge of the complete timestamp information. 

Assuming a reciprocal PA model, the empirical scenario proportions satisfy 
\begin{equation}
\label{eq:proportions}
\begin{split}
\hat{\alpha}_n &= \frac{1}{|E(n)|} \sum_{k = 1}^n \mathbf{1}_{\{J_k = 1\}} \xrightarrow{\text{a.s.}}  \frac{\alpha}{1 + \rho}, \\
\hat{\gamma}_n &= \frac{1}{|E(n)|} \sum_{k = 1}^n \mathbf{1}_{\{J_k = 3 \}} \xrightarrow{\text{a.s.}}  \frac{\gamma}{1 + \rho}, 
\end{split}
\end{equation}
where $\hat{\alpha}_n$ and $\hat{\gamma}_n$ are identifiable from the
\sid{edge list}, even when the timestamp is coarse.  
Considering a re-parametrization for the reciprocal PA model through $(\alpha,\beta,\iota,b)$, we give details on the extreme-value estimation approach.

By Theorem~\ref{thm:MRV}, we see that $(\mathcal{I},\mathcal{O})$ has
a distribution with standard regularly varying tails so that
$\mathcal{I}+\mathcal{O}$ has the same tail index $\iota$ as given in
\eqref{eq:MRV_IO}. 
Given in- and out-degree sequences \sid{$\{(\Din_v(n),\Dout_v(n)):
  v\in V(n)\}$}
from a reciprocal PA model with $n$
steps of evolution,  we apply
the minimum distance method (cf. \cite{clauset:shalizi:newman:2009})
to  
$\{\Din_v(n)+\Dout_v(n): v\in V(n)\}$ to obtain the tail index estimate $\widehat{\iota}$.
Further, recall from \eqref{e:defa} that 
the limit measure of the regularly varying measure $\PP[(\mathcal{I}, \mathcal{O})\in \cdot]$ concentrates on.
\begin{equation}
\label{eq:slope}
\left\lbrace (x, y) \in \mathbb{R}_+^2\setminus \{\origin\} : y = bx = \frac{\gamma-\alpha + \sqrt{D_0}}{2(\beta + \gamma)\rho} x \right\rbrace.
\end{equation}
We then estimate $b$ using the the angular density plot as follows. 

Using the selected threshold for $\{\Din_v(n)+\Dout_v(n): v\in V(n)\}$ produced by the minimum distance method along with necessary sanity checks using graphical
tools such as altHill plots (cf. \cite{drees:dehaan:resnick:2000}), we construct the
angular density plot for large in- and out-degree pairs, from which we find 
the estimate of the empirical mode, $\widehat{m}$.
We use the \verb6locmodes()6 function in R's \verb6multimode6 package to decide 
the bandwidth. 
Then it follows from the 
$L_1$-polar transform in \eqref{eq:angle} that
$$\widehat{b} =\frac{1}{\widehat{m}} - 1.$$
Equating the theoretical results in \eqref{eq:MRV_IO}, \eqref{eq:proportions} and \eqref{eq:slope} with their respective estimators $\hat{\iota}$, $\widehat{\alpha}_n$, $\widehat{\beta}_n$ and $\widehat{b}$, we solve a system of four equations with unknowns $(\alpha, \beta, \rho, \delta)$ to obtain the estimator $\widehat{\theta}_{EVT}$,
which is then compared with 
window estimators through simulation studies in Section~\ref{sec:sim}. Both estimation methods are also applied to a real dataset in Section~\ref{sec:real}.


\section{Simulation Studies}\label{sec:sim}

In this section, we apply the two estimation methods presented in Section \ref{sec:est} to simulated data.
We first assume the observed edge list is generated directly from the PA model with reciprocity so that the true value of window length is 0. Second, we assess the robustness of our estimation approaches by shifting the reciprocal edge added at step $k$ to step $k+W_k$, where $\{W_k:k\ge 1\}$ is a sequence of iid non-negative integer-valued random variables.


\subsection{Simulation of the PA model with Reciprocity}
\label{subsec:sim_recip}

We start by simulating 100 reciprocal PA networks where reciprocal edges are added at the same time step as their parent edge, and we set $\boldsymbol{\theta} = (0.2, 0.7, 0.2, 1)$ and $n = 10^5$. The resulting datasets conform completely with the assumptions in Theorem \ref{thm:MRV} and the clear correct choice of $w^\star$ is $0$. For the window estimators with optimal criteria in \eqref{eq:Wstar} and \eqref{eq:WstarTail}, we search for the optimal $w$ over the set $\mathcal{W} = \left\lbrace 0, 1, \dots, 50 \right\rbrace$. The absolute errors for the four parameters are displayed in the left panel of Figure \ref{fig:nopermutesims}, and 
we report the RMSE for both window and extreme-value based estimators in the first row of 
Table \ref{tab:RMSE}. 

\begin{table}[h]
\centering
\hspace{-1cm}
\begin{tabular}{c c c c}
\hline 
Law($W_k$) & $\text{RMSE}(\hat{\boldsymbol{\theta}}_{w_L^\star})$ (\%) & $\text{RMSE}(\hat{\boldsymbol{\theta}}_{w_T^\star})$ (\%) & $\text{RMSE}(\hat{\boldsymbol{\theta}}_\text{EVT})$ (\%)\\ 
\hline
$W_k =0$ & $(0.29, 0.14, 0.12, 2.82)$ & $(0.29, 0.14, 0.13, 2.87)$ & $(2.45, 3.97, 15.87, 31.99)$  \\
Geometric($0.3$) & $(0.15, 0.14, 0.14, 1.43)$ & $(0.57, 0.81, 3.30, 9.57)$  & $(2.83, 4.36, 17.39, 31.40)$ \\
\hline
\end{tabular}
\caption{RMSE (in \%) of $\hat{\boldsymbol{\theta}}_{w_L^\star}$, $\hat{\boldsymbol{\theta}}_{w_T^\star}$ and $\hat{\boldsymbol{\theta}}_{EVT}$ across 100 datasets of size $n = 10^5$ simulated from PA data across a variety of parameter choices and reciprocal edge shifts.}
\label{tab:RMSE}
\end{table}

\begin{figure}[h]
\centering
\includegraphics[scale = 0.45]{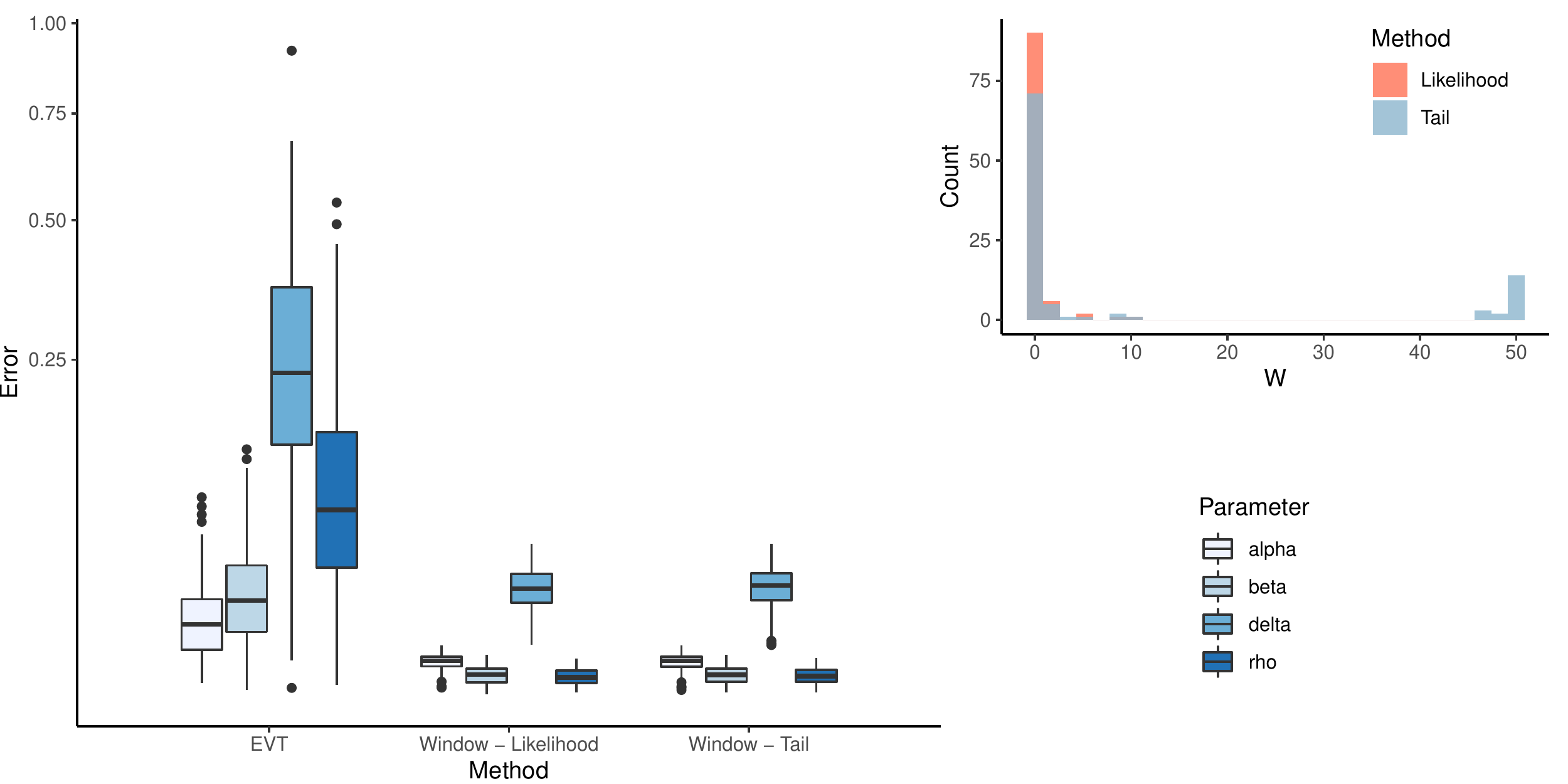}
\caption{Left: Absolute error of $\hat{\boldsymbol{\theta}}_{EVT}$, $\hat{\boldsymbol{\theta}}_{w_L^\star}$ and $\hat{\boldsymbol{\theta}}_{w_T^\star}$ (square-root scale) for 100 datasets simulated from a preferential attachment model with $\boldsymbol{\theta} = (0.2, 0.7, 0.2, 1)$ and $n = 10^5$.  
Right: The histogram for the optimal $w^\star$ chosen by \eqref{eq:Wstar} and \eqref{eq:WstarTail}. We allow $w$ to range from $0$ to $50$.}
\label{fig:nopermutesims}
\end{figure}

\begin{figure}[h]
\centering
\includegraphics[scale = 0.4]{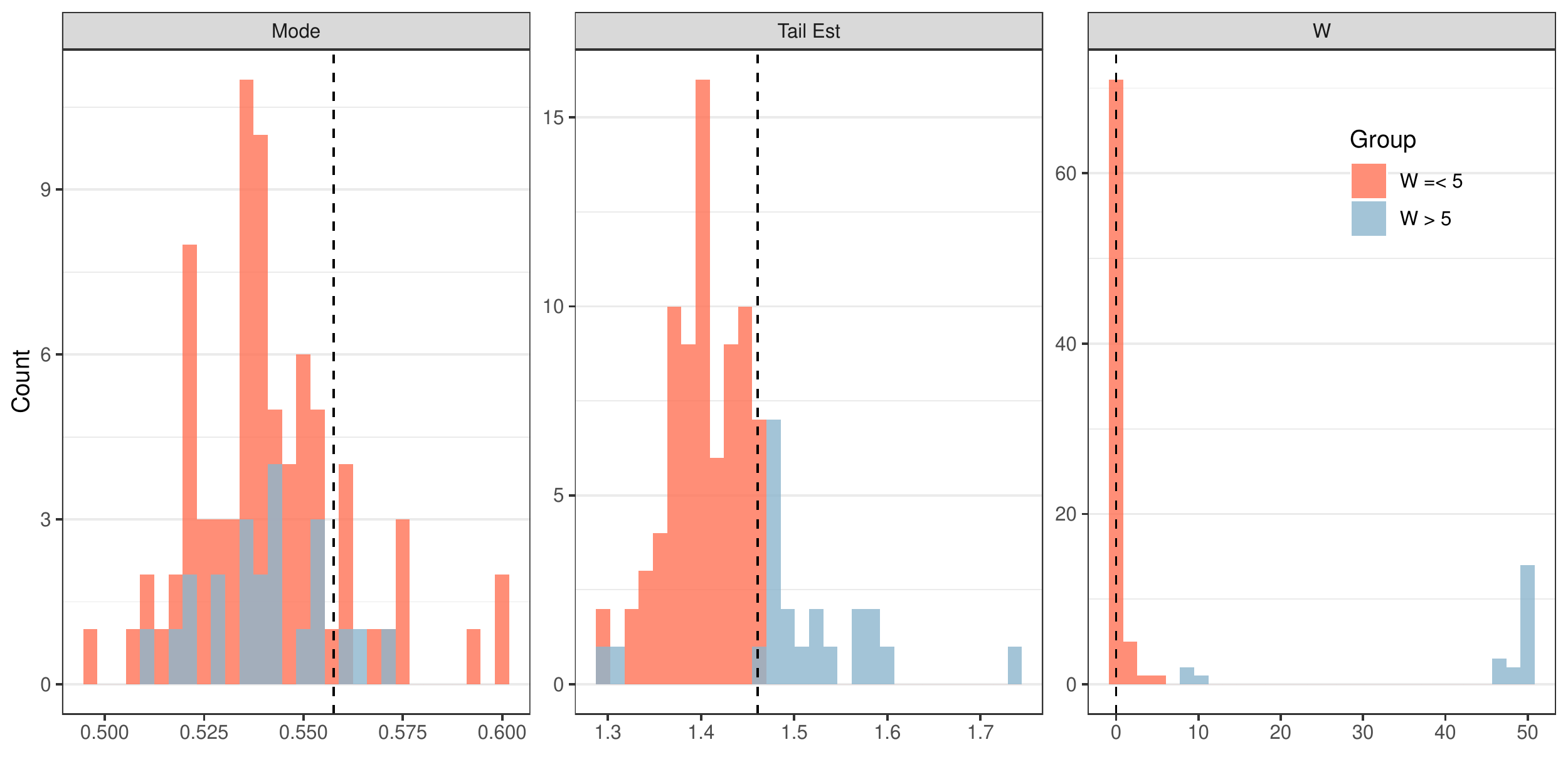}
\caption{Estimates of $m$, $\iota$ and the corresponding choice of $w^\star$ from \eqref{eq:Wstar} and \eqref{eq:WstarTail} for 100 datasets simulated from a preferential attachment model with $\boldsymbol{\theta} = (0.2, 0.7, 0.2, 1)$ and $n = 10^5$. Note that overestimates of $\iota$ often correspond with overestimates of $w_T^\star$.}
\label{fig:nopermutestats}
\end{figure}

As expected when the data comes from the actual model, both boxplots and RMSE values show that the window estimators outperform the extreme-value based estimator by having much lower absolute errors for all four parameters, regardless of the choice of criterion \eqref{eq:Wstar} or \eqref{eq:WstarTail}.
To understand why the extreme-value method performs poorly,
we track the estimated mode $\hat{m}$ and tail index $\hat{\iota}$ in
the left and middle panels of Figure~\ref{fig:nopermutestats},
respectively, and the vertical dashed lines correspond to the
theoretical values of $m$ and $\iota$. We see from these two panels
that the estimates are variable, and not always close to the
theoretical values.

\sid{There are several possible reasons for} the \sid{comparatively} poor performance \sid{of}
the extreme-value estimators.
According to the analysis in
\cite{drees:janssen:resnick:wang:2020}, even for iid data with
Pareto-like tails, the minimum distance method \sid{that we are
  relying on} tends to choose too
high a threshold. Simulation studies in
\cite{drees:janssen:resnick:wang:2020} for \sid{PA models} suggest
that 
performance of the minimum distance method in the standard PA regime
\sid{may depend} on the choice of parameters.   
Thus relying on the minimum distance method to determine the threshold
and therefore estimates of the slope and MRV index requires caution
and for real datasets, we need to visually consult Hill or altHill plots to further validate the threshold chosen by the minimum distance method. 
Additionally, unlike the window estimator which uses the entire network history, the
extreme-value approach makes inference  using  only a small
proportion of the data.

\sid{Features that recommend the extreme value method include} the fact that
when the timestamp information is coarse, i.e. multiple edges are
annotated with the same timestamp, one may still identify the
$\alpha$- and $\gamma$-scenarios so that the extreme-value approach
remains valid, whereas the proposed window method is not applicable.
Also, based on the evidence in \cite{wan:wang:davis:resnick:2017b}, 
when there is model error, the extreme-value method should
do better than the model dependent likelihood-based approach.

To compare the two window estimators, 
we plot a histogram for the chosen optimal $w^\star$ for each realization in the right panel of Figure \ref{fig:nopermutesims}, where 90 of the 100 trials give us an optimal $w_L^\star$ of $0$.
However, the selected $w_T^\star$ is not as accurate as $w_L^\star$, assigning only 71 of the 100 trials a $w_T^\star$ of $0$.
\sid{Figure \ref{fig:nopermutestats} highlights simulation} realizations with selected $w_T^\star>5$
and compares the corresponding
$\hat{m}$ and $\hat{\iota}$. We see from Figure
\ref{fig:nopermutestats} that the poor choices of $w_T^\star$ often
correspond to poor estimates of $\iota$ produced by the minimum
distance method. 
Despite the inaccuarcy of the optimal window length,
$\hat{\boldsymbol{\theta}}_{w_T^\star}$ performs as well as
$\hat{\boldsymbol{\theta}}_{w_L^\star}$ in terms of the RMSE,
indicating that the parameter estimates are not \sid{very} sensitive to the
choice of $w^\star$.  

\sid{We conclude that} the two criteria proposed in \eqref{eq:Wstar} and \eqref{eq:WstarTail} are both helpful tools to produce reasonable parameter estimates when reciprocal edges are created simultaneously.




\subsection{Robustness of Estimators Against Random Edge Shifts}

Based on the illustration in Table~\ref{tab:time_index} and the simulation results in the preceding section, we strongly suspect that the window estimators provide consistent estimates for the four parameters in the reciprocal PA model. 
However, suppose that we do not observe the reciprocal edge in $e_k$ until step $k+w_k$, and $w_k$ differs for each $k$. 
Will a constant optimal window length $w^\star$ be able to provide consistent parameter estimates?

To further examine the robustness of the proposed estimation procedures against random edge shifts, we assume the reciprocal edge created at step $k\in \mathcal{R}$ will not be observed until step $k+W_k$,
where $\{W_k:k\ge 1\}$ is a sequence of iid geometric random variables (also independent from the evolution of the reciprocal PA network) with pmf
\[
\PP(W_k=m) = 0.7^m 0.3,\qquad m=0,1,2,\ldots. 
\] 
Applying the two \sid{selection} criteria given in \eqref{eq:Wstar} and
\eqref{eq:WstarTail}, we compare the corresponding optimal window
lengths, as well as the performance of
$\hat{\boldsymbol{\theta}}_{w^\star_L}$ and
$\hat{\boldsymbol{\theta}}_{w^\star_T}$. 


\begin{figure}[h]
\centering
\includegraphics[scale = 0.45]{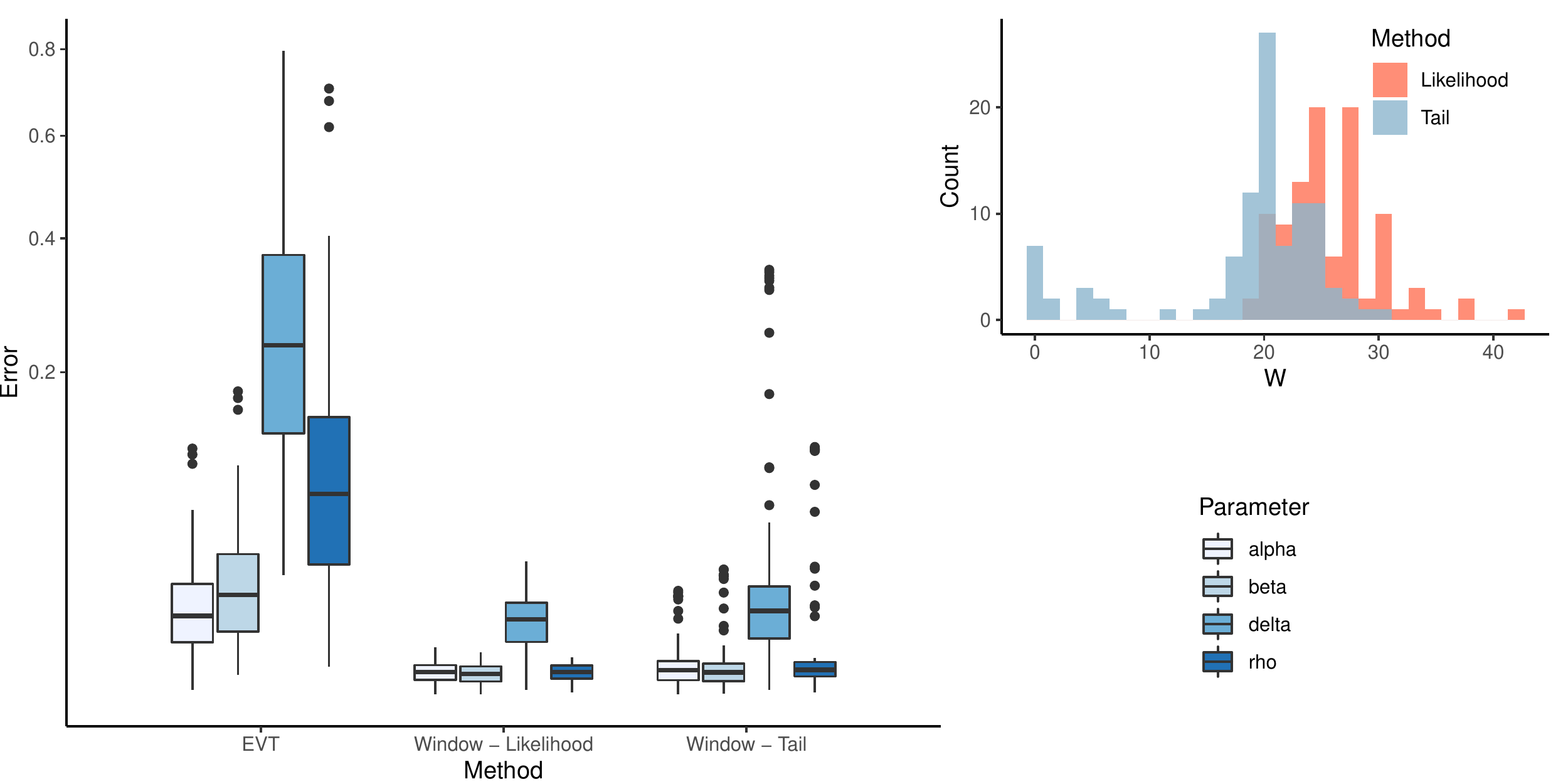}
\caption{Left: Absolute error of $\hat{\boldsymbol{\theta}}_{EVT}$, $\hat{\boldsymbol{\theta}}_{w^\star_L}$ and $\hat{\boldsymbol{\theta}}_{w^\star_T}$ (square-root scale) for 100 datasets simulated from the same setup as Figure \ref{fig:nopermutesims}, but now reciprocal edges are allocated at random times from birth according to a geometric distribution with mean $7/3$. 
Right: 
The histogram for the optimal $w^\star$ chosen by \eqref{eq:Wstar}. We allow $w$ to range from $0$ to $50$.}
\label{fig:permutesims}
\end{figure}

Keeping the same parameters $\boldsymbol{\theta} = (0.2, 0.7, 0.2, 1)$, $n=10^5$, and $\mathcal{W} = \left\lbrace 0, 1, \dots, 50 \right\rbrace$ as in the previous section, we 
apply both window (with two different optima\sid{lity} criteria) and extreme-value based approaches to 100 simulated reciprocal PA networks after random edge shifts. The absolute errors for the four parameters are displayed in the left panel of Figure \ref{fig:permutesims}, and 
the RMSE's for all three estimators are reported in the second row of 
Table \ref{tab:RMSE}. Similar to the results in
Figure~\ref{fig:nopermutesims},  the window estimators
 outperform the extreme-value based ones, and the RMSE's for the
window estimates remain small.  
Although the extreme-value estimates have larger RMSE's (especially for $\rho$ and $\delta$) compared to the window estimates, the random edge shifts do not degrade significantly the performance of the extreme-value estimators.
In contrast, even though
medians of the absolute errors for $\hat{\boldsymbol{\theta}}_{w_T^\star}$ and $\hat{\boldsymbol{\theta}}_{w_L^\star}$ are nearly identical, the absolute errors for $\hat{\boldsymbol{\theta}}_{w_T^\star}$ tend to take on extreme values more often.
The larger variation of $\hat{\boldsymbol{\theta}}_{w_T^\star}$ is presumably due to the variability of the estimated $\hat\iota$ by the minimum distance method as revealed in the middle panel of Figure~\ref{fig:nopermutestats}.


The right panel of Figure~\ref{fig:permutesims} gives the histogram for the selected $w^\star$ for all 100 realizations.
Although a reciprocal edge is, on average, added $7/3$ steps later than its parent edge, the choice of $w^\star$ determined by \eqref{eq:Wstar} and \eqref{eq:WstarTail}  is larger in order to capture the appropriate number of reciprocal edges that maximizes the likelihood for the relabeled data. 
Overall, selected window lengths, $w_L^\star$, are larger than $w_T^\star$, indicating that the criterion in \eqref{eq:Wstar} tends to favor a larger $\rho$. 
Figure \ref{fig:contour} displays a contour plot of $\ell\left((0.2, \beta, \rho, 1); (e_k)_{k = 1}^n\right)$ based on a single dataset $(e_k)_{k = 1}^n$ 
generated from the original reciprocal PA model without edge shifts. 
Then after applying the random edge shifts, we obtain the interpreted edge list using a variety of window lengths,
$w\in \mathcal{W}$, and 
black dots in Figure \ref{fig:contour} correspond to the log-likelihoods associated with
$\left\lbrace \hat{\boldsymbol{\theta}}_w: w \in \mathcal{W} \right\rbrace$. 
The estimated $\hat{\boldsymbol{\theta}}_{w^\star_L}$ with optimal $w_L^\star$ chosen by \eqref{eq:Wstar} is colored in red, and the true parameter $\boldsymbol{\theta}$ is colored in green. 
As the window length $w$ increases, 
we have a higher estimated $\hat{\rho}_w$ but a lower $\hat{\beta}_w$.
The window estimator $\hat{\boldsymbol{\theta}}_{w^\star_L}$ thus employs the likelihood based on the relabeled dataset $(e^w_k)_{k = 1}^n$ as a proxy for the likelihood of the original dataset $(e_k)_{k = 1}^n$ to optimally choose a $\hat{\boldsymbol{\theta}}_w$ near $\boldsymbol{\theta}$. 

\begin{figure}[h]
\centering
\includegraphics[scale = 0.45]{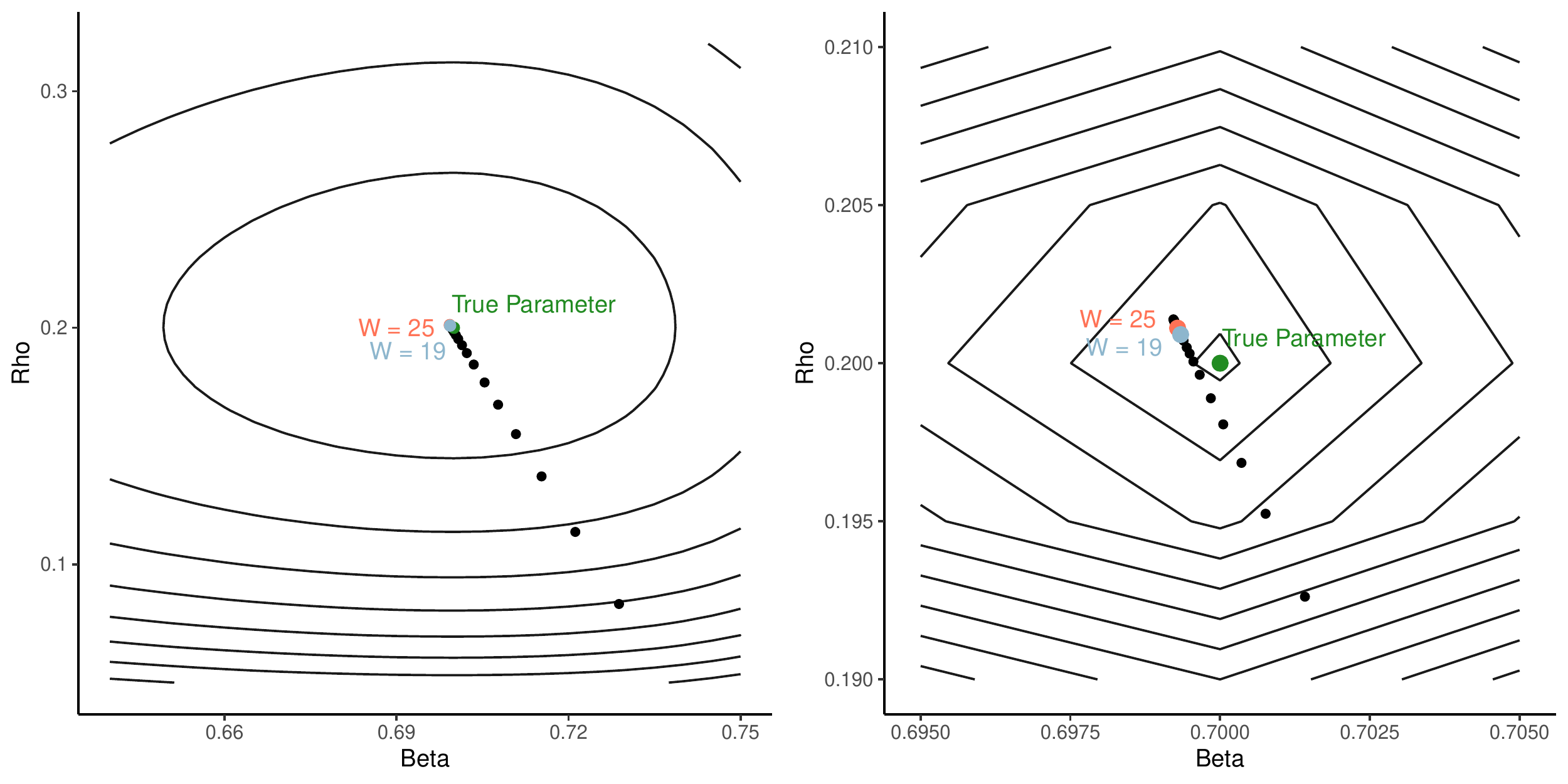}
\caption{Contour plot of $\ell\left((0.2, \beta, \rho, 1); (e_k)_{k = 1}^n\right)$ for a PA simulated dataset $(e_k)_{k = 1}^n$ with $\boldsymbol{\theta} = (0.2, 0.7, 0.2, 1)$ (green) and known events $(J_t)_{k = 1}^n$ and $(R_t)_{k =1}^n$. Plotted in black are estimates $\hat{\boldsymbol{\theta}}_w$
based on relabeled datasets $(e^{w}_k)_{k = 1}^n$ with $w^\star_L = 25$ (red) and $w^\star_T = 19$ (blue).}
\label{fig:contour}
\end{figure}

The good performance of $w^\star_L$ is further supported by the upper
panel of Figure \ref{fig:likelihood_comp}, where we record values of
$\ell\left(\hat{\boldsymbol{\theta}}_w; (e_k)_{k = 1}^n\right)$ and
$\ell\left(\hat{\boldsymbol{\theta}}_w; (e^w_k)_{k = 1}^n\right)$ as
$w$ increases and $\hat{\boldsymbol{\theta}}_w$ slices through the
parameter space. For  appropriately chosen $w$,
$\ell\left(\hat{\boldsymbol{\theta}}_w; (e^w_k)_{k = 1}^n\right)$ is a
good approximation of $\ell\left(\hat{\boldsymbol{\theta}}_w; (e_k)_{k
    = 1}^n\right)$. The chosen $w^\star_L$ is near the $w$ that
maximizes $\ell\left(\hat{\boldsymbol{\theta}}_w; (e_k)_{k =
    1}^n\right)$ (the green dot), which are 25 and 24,
respectively. In addition, beyond a $w$ of 17, both likelihoods are
relatively stable, indicating that the window method is also robust to
poor choices of $w$. As the window size increases, the estimates
$\hat{\boldsymbol{\theta}}_w$ coalesce in areas around the optimum
where the likelihood exhibits increased curvature. Hence, small
deviations in $w$ will not result in drastically different
$\hat{\boldsymbol{\theta}}_w$ estimates, 
which validates the application of the window estimation method to real datasets.

To check the performance of the \sid{selection} criterion in
\eqref{eq:WstarTail}, we use the same simulated dataset as in
Figure~\ref{fig:contour} and get $w^\star_T=19$. We mark the position
of $\hat{\boldsymbol{\theta}}_{w^\star_T}$ as the blue dot in the
contour plots of Figure~\ref{fig:contour}, 
and the blue dot in the upper panel of Figure~\ref{fig:likelihood_comp} gives the value of
$\ell\left(\hat{\boldsymbol{\theta}}_{w^\star_T}; (e^{w^\star_T}_k)_{k = 1}^n\right)$. Even though the selected $w^\star_T$ is smaller than $w^\star_L$, we do not observe much differences in the estimated parameters or the maximized likelihood. We also track values of $\hat{\iota}_w$ for different $w$ in the lower panel of Figure~\ref{fig:likelihood_comp}, and the tail estimate $\hat{\iota}$ chosen by the minimum distance method is given as the red line, which is slightly lower than the true value of $\iota$ (green line). Estimated $\hat{\iota}_{w^\star_L}$ and $\hat{\iota}_{w^\star_T}$ are colored in red and blue, respectively, both of which are close to $\iota$. Hence, we conclude that the criterion in \eqref{eq:WstarTail} provides a reasonable resolution in deciding the optimal window size, but its performance will depend on the accuracy of the minimum distance method.

\begin{figure}[h]
\centering
\includegraphics[scale = 0.45]{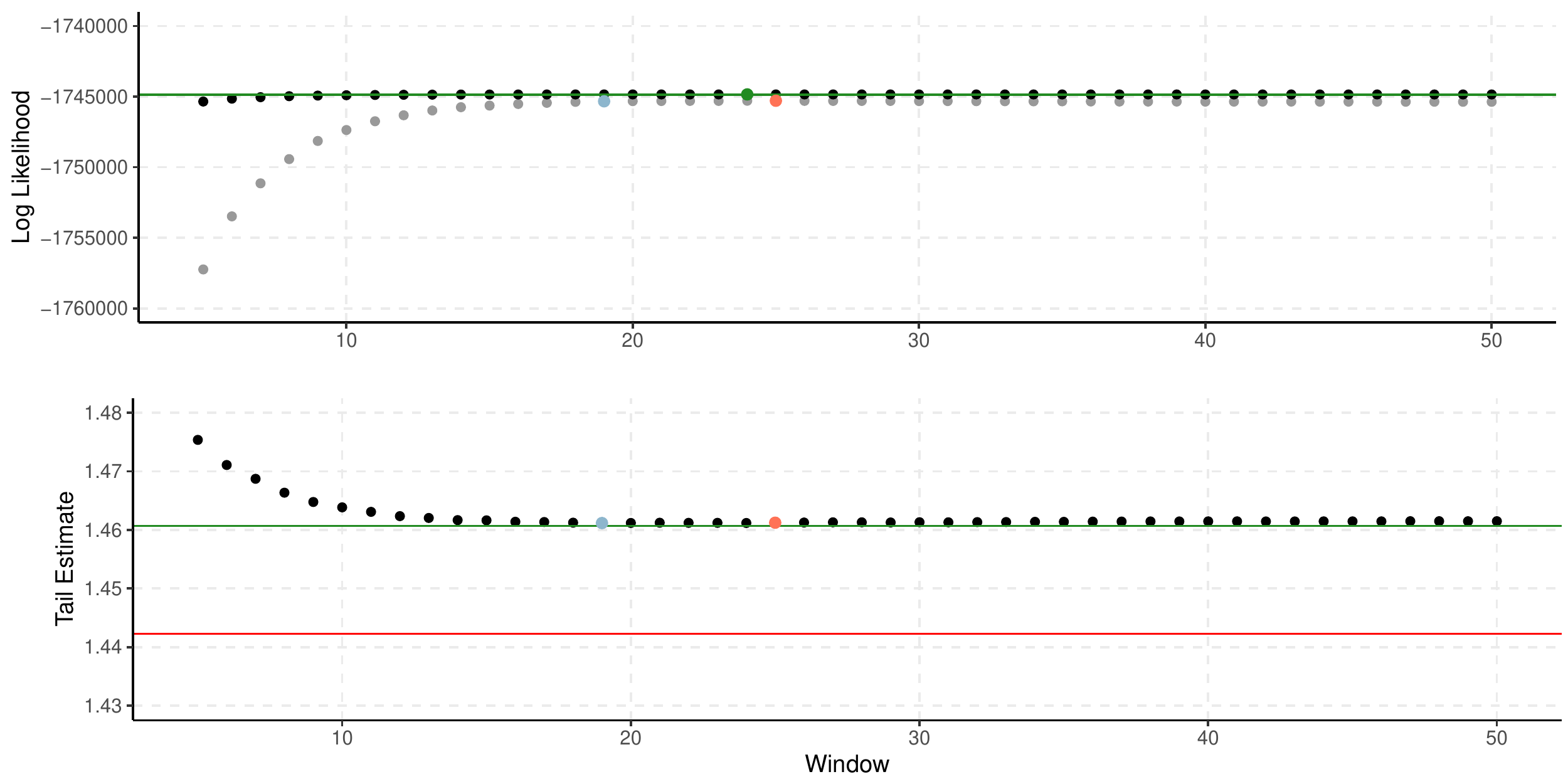}
\caption{Upper: Plot of $\ell\left(\hat{\boldsymbol{\theta}}_w; (e_k)_{k = 1}^n\right)$ (in black) and  $\ell\left(\hat{\boldsymbol{\theta}}_w; (e^w_k)_{k = 1}^n\right)$ (in gray). The selected $w^\star_L$ and $w^\star_T$ are colored in red and blue, respectively, while $\ell\left( \boldsymbol{\theta}; (e_k)_{k = 1}^n \right)$ is given by the green line. Lower: Plot of $\hat{\iota}_w$. The true value of $\iota$ is marked by the green line, and the tail estimate given by the minimum distance method is given by the red line. Values of $\hat{\iota}_{w^\star_L}$ and $\hat{\iota}_{w^\star_T}$ are colored in red and blue, respectively.}
\label{fig:likelihood_comp}
\end{figure}

\section{Real Data: Facebook Wall Posts}\label{sec:real}

\sid{This section demonstrates the mechanics} of fitting a reciprocal PA model to a
non-synthetic social network, namely the
Facebook wall post data on KONECT \footnote{The dataset is available
  at \url{http://konect.cc/networks/facebook-wosn-wall/}}
\cite{kunegis:2013}. \sid{In this dataset, the} network consists of 46,952 nodes and
876,993 edges where each node represents a user and each edge
represents a post to another user's wall. The users are based in New
Orleans and their posts are monitored over a period from 09/14/2004 to
01/22/2009. Note that user\sid{s} can post to their own wall and other
users' walls multiple times. 
We treat the edges added before April 7th, 2007 as the seed
graph. This date is chosen since it lies near the beginning of an
interval of time where the network experiences steady, almost linear
growth in the number of edges added from day to day, and we know from
\eqref{eq:proportions} that growth of the number of vertices and edges
should be linear. 
Details on different phases of growth for the Facebook data have been studied in \cite{wang:resnick:2019b}. 

\paragraph*{Data pre-processing.}
Here we also make two additional adjustments to the Facebook
data. (i) We remove nodes with in-degree 0 and out-degree greater
than 33 (the 80th percentile of nodes with in-degree 0). These nodes
\sid{exhibit} behavior that cannot be well-modeled by the reciprocal PA
model. \sid{Incorporating such nodes requires
the modeling of two  distinct populations with
different reciprocity features.} We leave the study of
modeling heterogeneous reciprocity levels in networks to future
work. (ii)  Several nodes with large in- and out-degrees in the
seed graph become inactive during our observation period \sid{from}
04/08/2007 to 01/22/2009. If we ignore their inactivity, then the
estimated $\hat{\iota}$ given by the minimum distance method will be
biased by the presence of these large but inactive nodes. Therefore,
we need to minimize the influence of the seed graph. When applying the
extreme-valued based approach, 
we delete nodes (together with their associated edges) in the seed graph which: (1) have not grown to at least 10 times their original size in the sum of in- and out-degrees and (2) have total degrees that are larger than the threshold selected by the minimum distance procedure. The latter adjustment is made since nodes with degrees that are below the threshold chosen by the minimum distance procedure will have no impact on the estimation of $\iota$ and $b$, assuming the same threshold is used for both estimators. We call the nodes that have undergone sufficient growth during the observation period active. Using the minimum distance method on the resulting $\Din_v(n) + \Dout_v(n)$ observations for active degrees corresponds to a tail index estimate of $\hat{\iota} = 1.309$.

\begin{figure}[h]
\centering
\includegraphics[scale = 0.5]{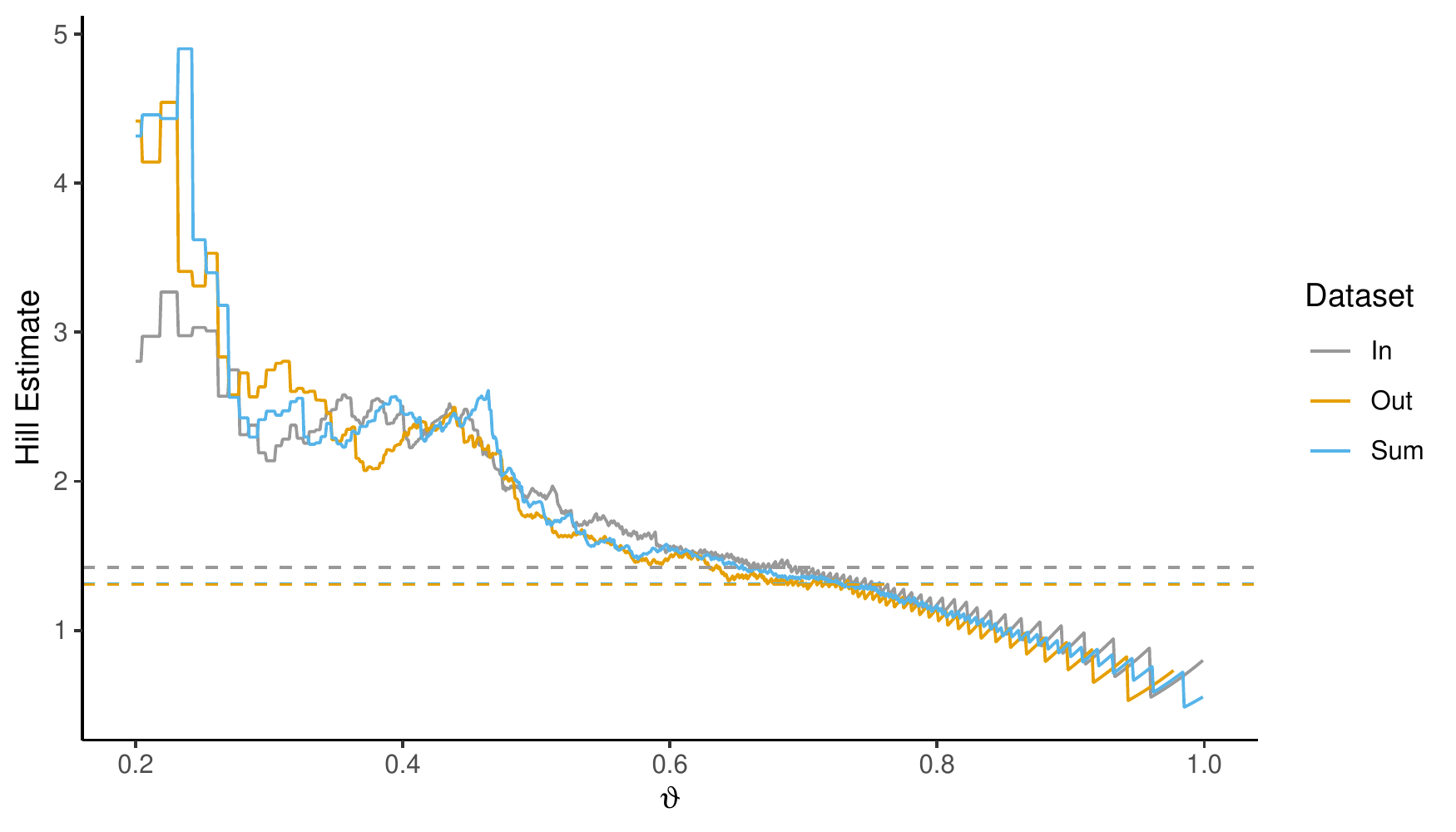}
\caption{AltHill plot of $\iota$. The dotted lines represents the empirical tail estimates of the in- and out-degrees along with their sum given by the minimum distance procedure. 
It is reasonable to assume the standard MRV result as in Theorem~\ref{thm:MRV}.}
\label{fig:AltHill}
\end{figure}

Before we fit the PA model with reciprocity, we first
check whether the standard MRV result given in Theorem~\ref{thm:MRV} is satisfied by the Facebook data. 
By \cite{drees:dehaan:resnick:2000}, one graphical tool to examine the appropriateness of 
the estimated $\hat{\iota}$ is the \emph{altHill plot},
which plots
$
\left\{(\theta, H^{-1}_{\lceil n^\vartheta\rceil, n}): 0\le \vartheta <1\right\}
$
with $H_{k, n}$ being the canonical Hill estimator using the upper
$k+1$ order statistics \citep{hill:1975}.
We give the altHill plot (cf. \cite{drees:dehaan:resnick:2000}) of $\hat{\iota}$ and the estimated in- and out-tail indices in Figure \ref{fig:AltHill}. 

Through visual inspections on the altHill plot, one might choose a threshold corresponding to $\vartheta \in (0.25, 0.5)$, but we find that this segment of order statistics mostly corresponds to nodes in the seed graph. 
In addition, the consistency result for Hill estimators based on undirected PA models \citep{wang:resnick:2018} requires the number of upper order statistics used being at least of order $O(\sqrt{n\log n})$.
Despite the lack of theoretical justifications for the reciprocal PA
model, we \sid{presume} such restriction is still needed and conclude that a
tail estimate corresponding to $\vartheta \in (0.25, 0.5)$ may not
\sid{accurately represent the asymptotic regime.}

The marginal in- and out-degree tail indices given by the minimum
distance procedure are 1.436 and 1.452 \sid{which is consistent with the}
standard MRV \sid{assumption}. Figure
\ref{fig:AltHill} also shows the the choice of threshold for
$\hat{\iota}$ is appropriate, as $\hat{\iota}$ estimates are stable
near the chosen threshold.

\paragraph*{Model fitting.}
We now apply the proposed estimation procedures to the pre-processed Facebook data. For the two window estimators, we search over $\mathcal{W} = \{1, 2, \dots, 300000 \}$. We return optimal window lengths of $w_L^\star = 222500$ and $w_T^\star = 219$ which correspond to an average real-time periods of 436 days and 8.5 hours, respectively. The choice of $w_L^\star$ is large since the likelihood function becomes rather flat beyond $w = 2000$, and the difference in the likelihood,
$\ell\left(\hat{\boldsymbol{\theta}}_{222500}; (e^{222500}_k)_{k = 1}^n\right)-\ell\left(\hat{\boldsymbol{\theta}}_{2000}; (e^{2000}_k)_{k = 1}^n\right)$, is small.
In addition, we expect the reciprocal posts on Facebook walls to happen within a short real-time interval, so we continue our analysis with the choice of $w_T^\star = 219$. 
The corresponding window estimates are
$$
\hat{\boldsymbol{\theta}}_{w_T^\star} = (\hat{\alpha}_{w_T^\star}, \hat{\beta}_{w_T^\star}, \hat{\rho}_{w_T^\star}, \hat{\delta}_{w_T^\star}) = (0.008, 0.952, 0.280, 4.247).
$$ 

To apply the extremes-value method, we first estimate $b$ by constructing the angular density for the in- and out- degrees of active nodes using the same threshold of  $\{ \Din_v(n) + \Dout_v(n) \}$ selected by the minimum distance procedure ($45$). This returns an estimate of $\hat{b} = 1.051$. Then solving for the unknown parameters gives 
$$
\hat{\boldsymbol{\theta}}_{EVT} = (\hat{\alpha}_{EVT}, \hat{\beta}_{EVT}, \hat{\rho}_{EVT}, \hat{\delta}_{EVT}) = (0.009, 0.953, 0.231, 7.385).
$$ 
Here both methods return similar estimates for $\alpha,\beta$, but
the window method gives a higher $\rho$ estimate than the extreme-value method.
Also, we see a large difference in the estimated values of $\delta$.
Here we speculatively attribute the small differences in $\alpha,\beta,\rho$ estimates to the robustness 
of the extreme-value method against modeling errors, since real datasets are never as clean as simulated ones where the ground truth is known. 

\begin{figure}[h]
\centering
\includegraphics[width = \textwidth]{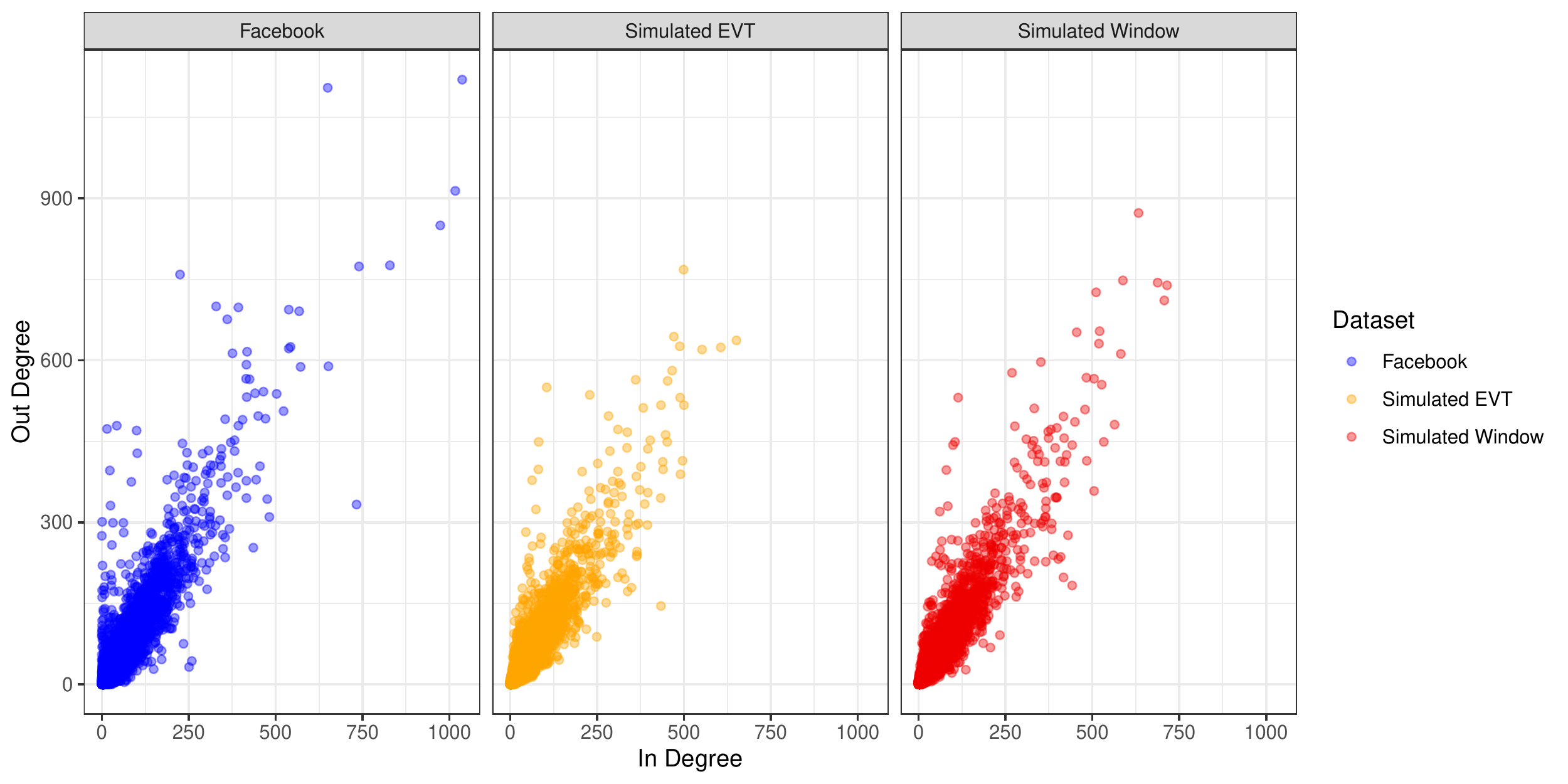}
\caption{Plots of in- and out-degrees for nodes from the pre-processed Facebook data (left) as well as PA-simulated datasets with parameters $\hat{\boldsymbol{\theta}}_{EVT}$ (middle) and $\hat{\boldsymbol{\theta}}_{w_L^\star}$ (right). The dependence between the in- and out- degrees is stronger for the window-simulated data in accordance with the higher estimate of $\rho$. }\label{fig:scatter_fb}
\end{figure}

Next, we assess the fit of the reciprocal PA model to the Facebook data based on the two methods. 
The middle and right panels of Figure~\ref{fig:scatter_fb} display two PA-simulated datasets of approximately the same number of edges as the Facebook data 
with parameters $\hat{\boldsymbol{\theta}}_{EVT}$ and $\hat{\boldsymbol{\theta}}_{w^\star_T}$, respectively. 
Note that both simulated datasets use the same initial graph resulting from accumulation until 04/07/2007.
For comparison, we also plot the pre-processed Facebook data in the left panel of Figure~\ref{fig:scatter_fb}.
The resulting graphs both capture the degree structure in the Facebook data, including the noteworthy dependence between in- and out- degrees. 
However, we also notice that the Facebook data still exhibits stronger concentration close to the y-axis, even after necessary pre-processing, compared to the two simulated datasets.
We speculate the potential existence of heterogeneous reciprocity levels for certain group of users,
e.g. different classes of users with different obsessive habits on social networks, which cannot be explained by the single parameter $\rho$.  
We leave the analysis of such heterogeneity as our future work.

\begin{figure}[h]
\centering
\includegraphics[width = \textwidth]{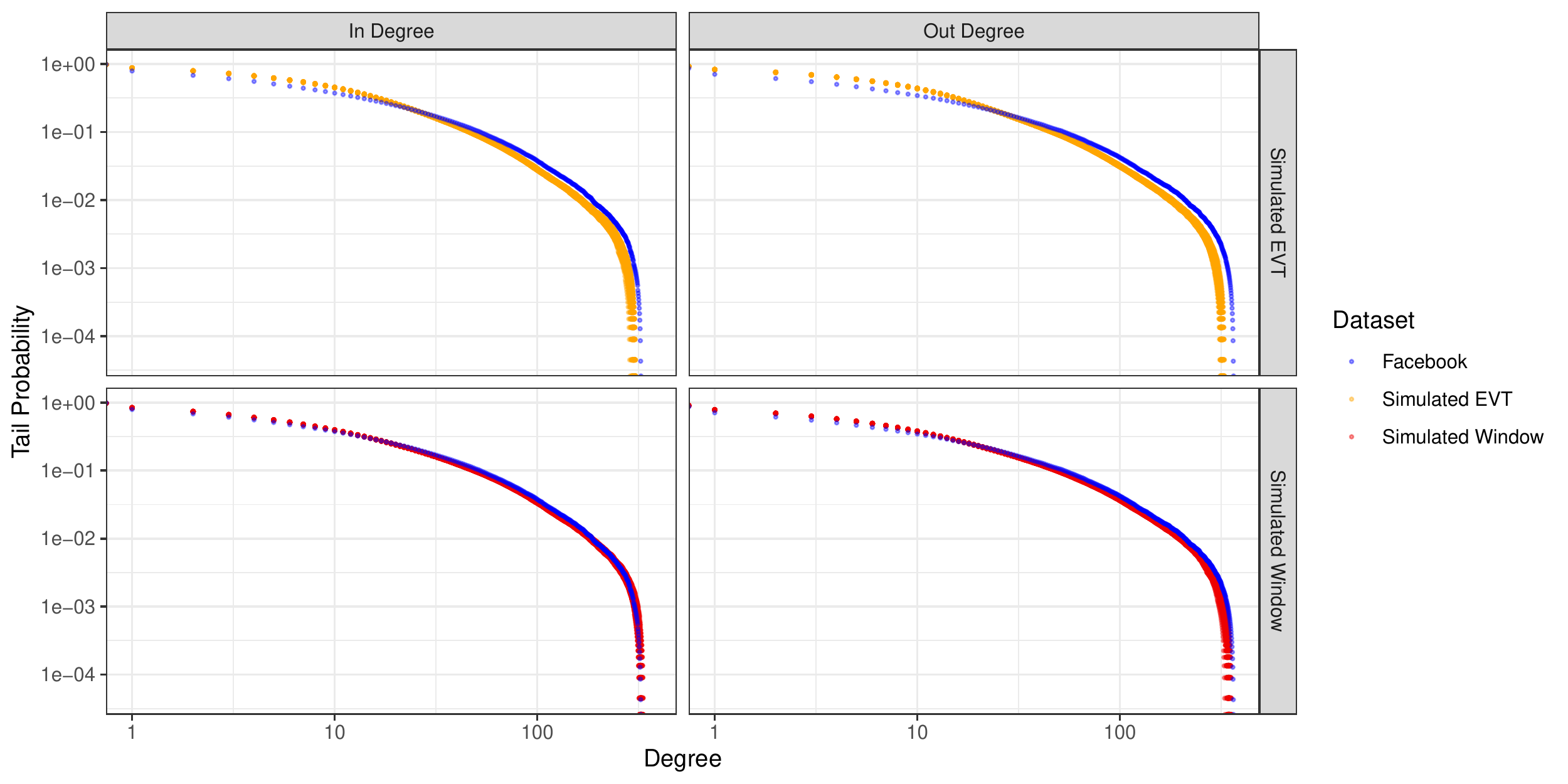}
\caption{The empirical marginal tail distributions for in- and out-degrees from the Facebook data (blue) and 50 PA-simulated datasets with parameters $\hat{\boldsymbol{\theta}}_{EVT}$ (yellow) and $\hat{\boldsymbol{\theta}}_{w_T^\star}$ (red). }
\label{fig:empiricaltail}
\end{figure}

In addition, for each set of estimated parameters $\hat{\boldsymbol{\theta}}_{EVT}$ and $\hat{\boldsymbol{\theta}}_{w^\star_T}$, we simulate 50 replications of the reciprocal PA model with the same seed graph, and examine the empirical marginal tail distributions for in- and out-degrees in Figure~\ref{fig:empiricaltail}. The blue dots represent the the tail distributions from the pre-processed Facebook data, and we see from Figure~\ref{fig:empiricaltail} that 
the window method provides a better fit in the marginal tails than the extreme-value method.
We also track the reciprocity coefficient for the Facebook graph along with 50 simulated datasets. The Facebook graph has a reciprocity coefficient of 0.711, while the simulated graphs with $\hat{\boldsymbol{\theta}}_{EVT}$ and $\hat{\boldsymbol{\theta}}_{w^\star_T}$ have average reciprocity coefficients of $0.440$ and $0.489$ with standard deviations of $0.00076$ and $0.00085$, respectively. 

We speculatively attribute the under-estimated reciprocity coefficient to the potential heterogeneous reciprocal behaviors for different classes of users. On one hand, 
users who always reply to wall posts may have higher personal
reciprocation levels than the estimated $\rho$ in our proposed model,
thus increasing the reciprocity coefficient for the entire graph. 
On the other hand, some users who always respond to their wall posts
may still log onto their Facebook account less frequently, so that
such reciprocal edges are not captured by the selected window, leading
to a lower estimated reciprocity coefficient.

Overall, despite the smaller reciprocity coefficients in the fitted
reciprocal PA models \sid{compared with the actual data},
the previously displayed simulated degree
distributions indicate that inclusion of the \sid{reciprocity feature in
  the model}  has \sid{enhanced}
the model's capability to accurately model \sid{evolution of in- and out-degree.}

\section{Concluding Remarks}\label{sec:rmk}

In this paper, we propose a PA model with reciprocity, which allows edge addition between two existing nodes. By embedding the in- and out-degree sequences into a sequence of multi-type branching processes with immigration,
we derive the convergence of the joint in- and out-degree counts, and study the asymptotic dependence between large in- and out-degrees. 
The allowance of adding edges between existing nodes leads to
challenges in the fitting of the proposed model. We propose two
different estimation approaches: (i) A window-based method which
makes use of the likelihood function, and (ii) an
extreme-value based method relying on the asymptotic dependence
structures between large in- and out-degrees.  
With full knowledge of the network evolution, the window method
provides more accurate parameter estimates, whereas the extreme-value
based method gives a \sid{solution} when the timestamp information is
coarse. 

When applied to the \sid{real} Facebook wall posts dataset,
both methods produce reasonable parameter estimates \sid{that}
capture the dependence between in- and out-degrees. 
However, \sid{the data analysis stimulates us }to speculate
about the existence of heterogeneous reciprocity levels among
different users.  
In future work, we plan to
assume a personalized $\rho$ parameter for each node in the network,
and discuss the theoretical properties as well as the model fitting. 

\bibliography{./bibfile_recip.bib}

\begin{appendix}
\section{Details on the Embedding Framework}\label{append:embed}
      
\sid{ Now we explain how to embed in- and out-degree sequences leading to
  Theorem \ref{thm:embed_MBI}. Assume $\{\bxi_{k,\delta} (\cdot),k\geq
  1\} $ are
  independent MBI processes with the same parameters but possibly
  different initializations.}
Suppose that $\bxi_{1,\delta}(0)=(1,1)$, and 
let $T_1$ be the first time when the $\bxi_{1,\delta}(\cdot)$ process jumps.
This can occur because $\xi^{(1)}_{1,\delta}(\cdot)$ jumps first, $\xi^{(2)}_{1,\delta}(\cdot)$ jumps first or because an immigration event occurs first. Therefore,
for $t\ge 0$, 
\begin{align*}
\PP\left(T_1>t\right)
&= \exp\left\{-t\left((\alpha+\beta)\bxi^{(1)}_{1,\delta}(0)+(\beta+\gamma)\bxi^{(2)}_{1,\delta}(0)
+(1+\beta)\delta\right)\right\}
=  e^{-t(1+\beta)(1+\delta)}.
\end{align*}
By Equations \eqref{eq:jump1} -- \eqref{eq:jump3}, we have
\begin{align*}
\PP\left(\bxi_{1,\delta}(T_1) = \bxi_{1,\delta}(0)+(1,0)\right)&=
(1-\rho)\frac{(\alpha+\beta)(\xi^{(1)}_{1,\delta}(0)+\delta)}{(1+\beta)(1+\delta)}
=(1-\rho)\frac{\alpha+\beta}{1+\beta},\\
\PP\left(\bxi_{1,\delta}(T_1) = \bxi_{1,\delta}(0)+(0,1)\right)&=
(1-\rho)\frac{(\beta+\gamma)(\xi^{(2)}_{1,\delta}(0)+\delta)}{(1+\beta)(1+\delta)}
= (1-\rho)\frac{\beta+\gamma}{1+\beta},\\
\PP\left(\bxi_{1,\delta}(T_1) = \bxi_{1,\delta}(0)+(1,1)\right)&=
\rho\left(\frac{(\alpha+\beta)\xi^{(1)}_{1,\delta}(0)}{(1+\beta)(1+\delta)}+
\frac{(\beta+\gamma)\xi^{(2)}_{1,\delta}(0)}{(1+\beta)(1+\delta)}
\right) + \rho \frac{(1+\beta)\delta}{(1+\beta)(1+\delta)}\\
&= \rho\left(\frac{(\alpha+\beta)(\xi^{(1)}_{1,\delta}(0)+\delta)}{(1+\beta)(1+\delta)}+
\frac{(\beta+\gamma)(\xi^{(2)}_{1,\delta}(0)+\delta)}{(1+\beta)(1+\delta)}\right)
=\rho.
\end{align*}
Additionally, if an immigration event of $(1,1)$ arrives at $T_1$, 
we assign it a type based on a coin flip. 
We label such a $(1,1)$ immigration event as type I with probability $(\alpha+\beta)/(1+\beta)$, and as type II with probability $1-(\alpha+\beta)/(1+\beta) = (\beta+\gamma)/(1+\beta)$.
Note also that at $T_1$, one of the following events happens: 
\begin{align*}
A_1(1)&:=\{\text{$\bxi_{1,\delta}(\cdot)$ is augmented by $(1,0)$ at $T_1$}\},\\
A_2(1)&:=\{\text{$\bxi_{1,\delta}(\cdot)$ is augmented by $(0,1)$ at $T_1$}\},\\
M_1(1) &:= \{\text{A type 1 particle in $\bxi_{1,\delta}(\cdot)$ splits into two type 1 particles and one type 2 particle at $T_1$}\}\\
&\quad\cup\{\text{A type I immigration event of $(1,1)$ arrives at $T_1$}\},\\
M_2(1) &:= \{\text{A type 2 particle in $\bxi_{1,\delta}(\cdot)$ splits into two type 2 particles and one type 1 particle at $T_1$}\}\\
&\quad\cup\{\text{A type II immigration event of $(1,1)$ arrives at $T_1$}\}.
\end{align*}
Here events $A_1(1), A_2(1)$ correspond to edge addition at step 1 of the network evolution, and $M_1(1), M_2(1)$ correspond to the creation of reciprocal/mutual edges at step 1.

Let $Z^{(1)}_1$ and $Z^{(2)}_1$ be two independent Bernoulli random variables, which are also independent from the MBI processes $\{\bxi_{k,\delta}(\cdot)\}$, and satisfy 
\begin{align}
\label{eq:prob_Z}
\PP(Z^{(1)}_1 = 1) &= \frac{\alpha(1+\beta)}{\alpha+\beta}=1-\PP(Z^{(1)}_1 = 0),\quad
\PP(Z^{(2)}_1 = 1) = \frac{\gamma(1+\beta)}{\beta+\gamma}=1-\PP(Z^{(2)}_1 = 0).
\end{align}
In order to decide whether a new MBI process needs to be initiated at $T_1$, we introduce a 
Bernoulli random variable $B_1$ such that
\[
B_1 :=  Z^{(1)}_1\left(\ind_{A_1(1)}+\ind_{M_1(1)}\right) + Z^{(2)}_1\left(\ind_{A_2(1)}+\ind_{M_2(1)}\right),
\]
which gives
\begin{align}
\PP&\left(B_1=1\middle|A_1(1)\right)
=\PP\left(B_1=1\middle|M_1(1)\right)
= \frac{\alpha(1+\beta)}{\alpha+\beta}\nonumber\\
& =1- \PP\left(B_1=0\middle| M_1(1)
\right)=1- \PP\left(B_1=0\middle| A_1(1)\right),\label{eq:B1_alpha}\\
\PP&\left(B_1=1\middle| A_2(1)\right)=\PP\left(B_1=1\middle|M_2(1)\right)
= \frac{\gamma(1+\beta)}{\beta+\gamma}\nonumber\\
&= 1- \PP\left(B_1=0\middle| M_2(1)\right)
 = 1- \PP\left(B_1=0\middle| A_2(1)\right).
\label{eq:B1_gamma}
\end{align}
Combining \eqref{eq:B1_alpha} and \eqref{eq:B1_gamma} shows that
$\PP(B_1=1)=1-\beta = 1-\PP(B_1=0)$. 
If $B_1=1$, we will initiate a new MBI process at $T_1$, corresponding to a new node being added.
Also, $\{B_1=0\}$ corresponds to the scenario where we add a self loop for node 1 at step 1, and no new node is created.

We now have the following cases to consider:
\begin{enumerate}
\item If event $A_1(1)$ happens at $T_1$ and $B_1 =1$, then 
we initiate the process, $\{\bxi_{2,\delta}(t-T_1):t\ge T_1\}$ with 
$\bxi_{2,\delta}(0)=(0,1)$.
\item If event $A_2(1)$ happens at $T_1$ and $B_1 =1$, then 
we initiate the process, $\{\bxi_{2,\delta}(t-T_1):t\ge T_1\}$ with 
$\bxi_{2,\delta}(0)=(1,0)$.
\item If either $M_1(1)$ or $M_2(1)$ happens at $T_1$ and $B_1 =1$, then 
we initiate the process, $\{\bxi_{2,\delta}(t-T_1):t\ge T_1\}$ with 
$\bxi_{2,\delta}(0)=(1,1)$.
\item If event $A_1(1)$ happens at $T_1$ and $B_1 =0$, then 
we do not count all following jumps of $\bxi_{1,\delta}(\cdot)$ until $\bxi_{1,\delta}(\cdot)$ is
increased by $(0,1)$.
\item If event $A_2(1)$ happens at $T_1$ and $B_1 =0$, then 
we do not count all following jumps of $\bxi_{1,\delta}(\cdot)$ until $\bxi_{1,\delta}(\cdot)$ is
increased by $(1,0)$.
\item If event $M_1(1)$ happens at $T_1$ and $B_1 =0$, then 
we do not count all following jumps of $\bxi_{1,\delta}(\cdot)$ until one type 2 particle in $\bxi_{1,\delta}(\cdot)$ splits into two type 2 particles and one type 1 particle, or a type II immigration event of $(1,1)$ arrives.
\item If event $M_2(1)$ happens at $T_1$ and $B_1 =0$, then 
we do not count all following jumps of $\bxi_{1,\delta}(\cdot)$ until one type 1 particle in $\bxi_{1,\delta}(\cdot)$ splits into two type 1 particles and one type 2 particle, or a type I immigration event of $(1,1)$ arrives.
\end{enumerate}
The first three cases show that when the new MBI process $\{\bxi_{2,\delta}(t-T_1):t\ge T_1\}$ is initiated, its initial value, $\bxi_{2,\delta}(0)$, satisfies
\begin{align*}
\EE\left(\bs^{\bxi_{2,\delta}(0)}\middle|B_1=1 \right) = \frac{\alpha(1-\rho)}{\alpha+\gamma}s_2 + \frac{\gamma(1-\rho)}{\alpha+\gamma}s_1 +\rho s_1s_2,\qquad \bs\in [0,1]^2.
\end{align*}
We set $\bxi_{2,\delta}(0)=(0,0)$ if $B_1=0$, then $\bxi_{2,\delta}(0)$ is a random vector with generating function
\begin{align}
\label{eq:pgfxi0}
\EE\left(\bs^{\bxi_{2,\delta}(0)}\right) = {\alpha(1-\rho)}s_2 + {\gamma(1-\rho)}s_1 
+(\alpha+\gamma)\rho s_1s_2 +\beta,\qquad \bs\in [0,1]^2.
\end{align}
Set $T_1^*=T_1$ when $B_1=1$.

When $B_1=0$, we
let $T^*_1$ be the first time after $T_1$ when a jump of the desired type occurs.
For case 4, we see that 
\begin{align*}
&\PP\left(A_1(1), B_1=0,\bxi_{1,\delta}(T^*_1) = \bxi_{1,\delta}(0)+(1,0) + (0,1)\right)\\
&= \frac{(1-\alpha)\beta}{\alpha+\beta}(1-\rho)\frac{(\alpha+\beta)(\xi^{(1)}_{1,\delta}(0)+\delta)}{(1+\beta)(1+\delta)}\\
&\qquad\times\sum_{k=0}^\infty\left(1-(1-\rho)\frac{(\beta+\gamma)(\xi^{(2)}_{1,\delta}(0)+\delta)}{(1+\beta)(1+\delta)}\right)^k (1-\rho)\frac{(\beta+\gamma)(\xi^{(2)}_{1,\delta}(0)+\delta)}{(1+\beta)(1+\delta)}\\
&= \frac{(1-\alpha)\beta}{1+\beta}(1-\rho).
\end{align*}
Similarly, we have for case 5,
\begin{align*}
&\PP\left(A_2(1), B_1=0,\bxi_{1,\delta}(T^*_1) = \bxi_{1,\delta}(0)+(0,1) + (1,0)\right)\\
&= \frac{(1-\gamma)\beta}{\beta+\gamma}(1-\rho)\frac{(\beta+\gamma)(\xi^{(2)}_{1,\delta}(0)+\delta)}{(1+\beta)(1+\delta)}\\
&\qquad\times\sum_{k=0}^\infty\left(1-(1-\rho)\frac{(\alpha+\beta)(\xi^{(1)}_{1,\delta}(0)+\delta)}{(1+\beta)(1+\delta)}\right)^k (1-\rho)\frac{(\alpha+\beta)(\xi^{(1)}_{1,\delta}(0)+\delta)}{(1+\beta)(1+\delta)}\\
&= \frac{(1-\gamma)\beta}{1+\beta}(1-\rho).
\end{align*}
Hence, combining cases 4 and 5 gives
\[
\PP\left(B_1=0,\bxi_{1,\delta}(T^*_1) = \bxi_{1,\delta}(0)+(0,1) + (1,0)\right)
=\beta(1-\rho),
\]
which corresponds to adding a self loop for node 1 at step 1 under the $\beta$-scenario with no reciprocal edge.
Following a similar reasoning, we have for cases 6 and 7 that
\begin{align*}
&\PP\left(M_1(1), B_1=0,\bxi_{1,\delta}(T^*_1) = \bxi_{1,\delta}(0)+(1,1) + (1,1)\right)\\
&= \frac{(1-\alpha)\beta}{\alpha+\beta}\rho\frac{(\alpha+\beta)(\xi^{(1)}_{1,\delta}(0)+\delta)}{(1+\beta)(1+\delta)}\sum_{k=0}^\infty\left(1-\rho\frac{(\beta+\gamma)(\xi^{(2)}_{1,\delta}(0)+\delta)}{(1+\beta)(1+\delta)}\right)^k \rho\frac{(\beta+\gamma)(\xi^{(2)}_{1,\delta}(0)+\delta)}{(1+\beta)(1+\delta)}\\
&= \frac{(1-\alpha)\beta}{1+\beta}\rho,
\end{align*} 
and 
\begin{align*}
&\PP\left(M_2(1), B_1=0,\bxi_{1,\delta}(T^*_1) = \bxi_{1,\delta}(0)+(1,1) + (1,1)\right)\\
&=\frac{(1-\gamma)\beta}{\beta+\gamma}\rho\frac{(\beta+\gamma)(\xi^{(2)}_{1,\delta}(0)+\delta)}{(1+\beta)(1+\delta)}\sum_{k=0}^\infty\left(1-\rho\frac{(\alpha+\beta)(\xi^{(1)}_{1,\delta}(0)+\delta)}{(1+\beta)(1+\delta)}\right)^k \rho\frac{(\alpha+\beta)(\xi^{(1)}_{1,\delta}(0)+\delta)}{(1+\beta)(1+\delta)}\\
&= \frac{(1-\gamma)\beta}{1+\beta}\rho.
\end{align*}
Therefore, combining cases 6 and 7 gives
\[
\PP\left(B_1=0,\bxi_{1,\delta}(T^*_1) = \bxi_{1,\delta}(0)+(1,1) + (1,1)\right)
=\beta\rho,
\]
which corresponds to adding two self loops for node 1 at step 1, under the $\beta$-scenario with a reciprocal edge created.

When $B_1=0$, let $T_1'$ denote the last time before $T_1^*$ when a jump with an undesired type occurs, then for $\Delta T_1:= (1-B_1)(T_1^*-T_1')$, we have
\[
\PP(\Delta T_1>t|B_1=0) = \left(\frac{(\beta+\gamma)}{1+\beta}(1-\rho)+\frac{(\alpha+\beta)}{1+\beta}(1-\rho)+\rho\right)e^{-(1+\beta)(1+\delta)t}
= e^{-(1+\beta)(1+\delta)t}.
\]
Since we do not count all jumps until the desired type of jump appears, then up to time $T_1^*$, the effective amount of evolution time for $\bxi_{1,\delta}$ is $T_1+ \Delta T_1$, and
\[
\EE(T_1+ \Delta T_1) = \frac{1+\beta}{(1+\beta)(1+\delta)} = \frac{1}{1+\delta}.
\]
Define $R_1:= \ind_{\{\bxi_{1,\delta}(T_1)=\bxi_{1,\delta}(0)+(1,1)\}}$, then
by cases 1--7, we see that
 $\PP(R_1=1)=\rho=1-\PP(R_1=0)$.
With $T^*_0=0$, we also set
$$
\mathcal{F}_{T^*_1}:= \sigma\left(B_1,R_1; \left\{\bxi_{k,\delta}(t-T^*_{k-1}):t\in [T^*_{k-1},T^*_1]\right\}_{k=1,1+B_1}\right).
$$

Set $S_1=0$, and $S_k:= \min\left\{s\ge 1: 1+\sum_{l=1}^s B_l = k\right\}$, $k\ge 2$,
then $T^*_{S_k}$
denotes the birth time of the $k$-th MBI process.
In general, for $n\ge 1$, suppose that we have initiated $N(n):=1+\sum_{k=1}^n B_k$ MBI processes at time $T^*_n$, i.e.
\begin{align}\label{eq:MBIs}
\{\bxi_{k,\delta}(t-T^*_{S_k}): t\ge T^*_{S_k}\}_{1\le k\le N(n)}.
\end{align}
Let $T_{n+1}$ be the first time after $T_n^*$ when one of the processes in \eqref{eq:MBIs} jumps,
and 
$$R_{n+1}:= \ind_{\left\{\bxi_{k,\delta}(T_{n+1}-T^*_{S_k})=\bxi_{k,\delta}(T^*_{n}-T^*_{S_k})+(1,1),\,\text{for some } 1\le k\le N(n)\right\}}.$$
Define the $\sigma$-algebra:
\begin{align*}
\mathcal{F}_{T^*_n}&:= \sigma\left(\{B_l,R_l: l=1,\ldots,n\}; \left\{\bxi_{k,\delta}(t-T^*_{k-1}):t\in [T^*_{k-1},T^*_n]\right\}_{k=1,\ldots, N(n)}\right),
\end{align*}
and we have for $t\ge 0$,
\begin{align*}
\PP^{\mathcal{F}_{T_n^*}}(T_{n+1}-T_n^*>t) &= 
e^{-t\left[(\alpha+\beta)\sum_{k=1}^{N(n)}\xi^{(1)}_{k,\delta}(T_n^*-T_{S_k}^*)
+(\beta+\gamma)\sum_{k=1}^{N(n)}\xi^{(2)}_{k,\delta}(T_n^*-T_{S_k}^*) + (1+\beta)\delta N(n)\right]}\\
&= e^{-t(1+\beta)\left(n+1+\sum_{k=1}^n R_k+\delta N(n)\right)}.
\end{align*}

At $T_{n+1}$, one of the following events takes place:
\begin{align*}
A_1(n+1)&:=\{\text{One of $\{\bxi_{k,\delta}(\cdot): 1\le k\le N(n)\}$ is augmented by $(1,0)$ at $T_{n+1}$}\},\\
A_2(n+1)&:=\{\text{One of $\{\bxi_{k,\delta}(\cdot): 1\le k\le N(n)\}$ is augmented by $(0,1)$ at $T_{n+1}$}\},\\
M_1(n+1) 
&:= \{\text{A type 1 particle in $\{\bxi_{k,\delta}(\cdot): 1\le k\le N(n)\}$ splits into}\\
&\qquad\qquad \text{two type 1 particles and one type 2 particle at $T_{n+1}$}\}\\
&\quad\cup\{\text{A type I immigration event of $(1,1)$ arrives at $T_{n+1}$}\},\\
M_2(n+1) &:= \{\text{A type 2 particle in $\{\bxi_{k,\delta}(\cdot): 1\le k\le N(n)\}$ splits into}\\
&\qquad\qquad \text{two type 2 particles and one type 1 particle at $T_{n+1}$}\}\\
&\quad\cup\{\text{A type II immigration event of $(1,1)$ arrives at $T_{n+1}$}\}.
\end{align*}
Here events $A_1(n+1), A_2(n+1)$ correspond to edge addition at step $n+1$ of the network evolution, and $M_1(n+1), M_2(n+1)$ correspond to the creation of reciprocal/mutual edges at step $n+1$.
Let $\{Z^{(1)}_k:k\ge 1\}$ and $\{Z^{(2)}_k:k\ge 1\}$ be two independent sequences of iid Bernoulli random variables, which are also independent from the MBI processes $\{\bxi_{k,\delta}(\cdot)\}$, and satisfy \eqref{eq:prob_Z}.
Define 
\[
B_{n+1} = Z^{(1)}_{n+1}\left(\ind_{A_1(n+1)}+\ind_{M_1(n+1)}\right)
+Z^{(2)}_{n+1}\left(\ind_{A_2(n+1)}+\ind_{M_2(n+1)}\right),
\]
and we have
\begin{align}
\PP^{\mathcal{F}_{T_n^*}}&\left(B_{n+1}=1\middle|A_1(n+1)\right)
=\PP^{\mathcal{F}_{T_n^*}}\left(B_{n+1}=1\middle|M_1(n+1)\right)
= \frac{\alpha(1+\beta)}{\alpha+\beta} \nonumber\\
&=1- \PP^{\mathcal{F}_{T_n^*}}\left(B_{n+1}=0\middle| M_1(n+1)
\right)
=1- \PP^{\mathcal{F}_{T_n^*}}\left(B_{n+1}=0\middle| A_1(n+1)\right),\label{eq:Bn_alpha}\\
\PP^{\mathcal{F}_{T_n^*}}&\left(B_{n+1}=1\middle| A_2(n+1)\right)=\PP^{\mathcal{F}_{T_n^*}}\left(B_{n+1}=1\middle|M_2(n+1)\right)
= \frac{\gamma(1+\beta)}{\beta+\gamma}\nonumber\\
&= 1- \PP^{\mathcal{F}_{T_n^*}}\left(B_{n+1}=0\middle| M_2(n+1)\right)
 = 1- \PP^{\mathcal{F}_{T_n^*}}\left(B_{n+1}=0\middle| A_2(n+1)\right).
\label{eq:Bn_gamma}
\end{align}
Since for $n\ge 1$, we have
\begin{align*}
\PP^{\mathcal{F}_{T_n^*}}(A_1(n+1))&=(\alpha+\beta)(1-\rho)/(1+\beta),
\qquad \PP^{\mathcal{F}_{T_n^*}}(A_2(n+1))=(\beta+\gamma)(1-\rho)/(1+\beta),\\
\PP^{\mathcal{F}_{T_n^*}}(M_1(n+1))&=(\alpha+\beta)\rho/(1+\beta),
\qquad \PP^{\mathcal{F}_{T_n^*}}(M_2(n+1))=(\beta+\gamma)\rho/(1+\beta),
\end{align*}
then it follows from \eqref{eq:Bn_alpha} and \eqref{eq:Bn_gamma} that
$\PP(B_{n+1}=1)= 1-\beta = 1-\PP(B_{n+1}=0)$. Also, \eqref{eq:Bn_alpha} and \eqref{eq:Bn_gamma} show that $B_{n+1}$ is independent from $\mathcal{F}_{T_n^*}$, i.e. for $b_k\in \{0,1\}$, $1\le k\le n+1$,
\begin{align}
\PP\left(B_k=b_k, k=1,\ldots, n+1\right) &= \EE\left(\ind_{\{B_{n+1}=b_{n+1}\}}\PP^{\mathcal{F}_{T_n^*}}\left(B_k=b_k, k=1,\ldots, n\right)\right)\nonumber\\
&=\cdots = \prod_{k=1}^{n+1} \PP\left(B_k=b_k\right).\label{eq:prob_Bk}
\end{align}
 If $B_{n+1}=1$, we will initiate a new MBI process at $T_{n+1}$, and
$\{B_{n+1}=0\}$ corresponds to the $\beta$-scenario, where no new node is created at step $n+1$.

Similar to the $n=1$ case, we consider the following scenarios:
\begin{enumerate}
\item If event $A_1(n+1)$ happens at $T_{n+1}$ and $B_{n+1} =1$, then 
we initiate the process, $\{\bxi_{N(n)+1,\delta}(t-T_{n+1}):t\ge T_{n+1}\}$ with 
$\bxi_{N(n)+1,\delta}(0)=(0,1)$.
\item If event $A_2(n+1)$ happens at $T_{n+1}$ and $B_{n+1} =1$, then 
we initiate the process, $\{\bxi_{N(n)+1,\delta}(t-T_{n+1}):t\ge T_{n+1}\}$ with 
$\bxi_{N(n)+1,\delta}(0)=(1,0)$.
\item If either $M_1(n+1)$ or $M_2(n+1)$ happens at $T_{n+1}$ and $B_{n+1} =1$, then 
we initiate the process, $\{\bxi_{N(n)+1,\delta}(t-T_{n+1}):t\ge T_{n+1}\}$ with 
$\bxi_{N(n)+1,\delta}(0)=(1,1)$.
\item If event $A_1(n+1)$ happens at $T_{n+1}$ and $B_{n+1} =0$, then 
we do not count all following jumps of the MBI processes in \eqref{eq:MBIs} until one of them is
increased by $(0,1)$.
\item If event $A_2(n+1)$ happens at $T_{n+1}$ and $B_{n+1} =0$, then 
we do not count all following jumps of the MBI processes in \eqref{eq:MBIs} until one of them is
increased by $(1,0)$.
\item If event $M_1(n+1)$ happens at $T_{n+1}$ and $B_{n+1} =0$, then 
we do not count all following jumps of the MBI processes in \eqref{eq:MBIs} until one type 2 particle splits into two type 2 particles and one type 1 particle, or a $(1,1)$ immigration event of type II arrives.
\item If event $M_2(n+1)$ happens at $T_{n+1}$ and $B_{n+1} =0$, then 
we do not count all following jumps of the MBI processes in \eqref{eq:MBIs} until one type 1 particle splits into two type 1 particles and one type 2 particle, or a $(1,1)$ immigration event of type I arrives.
\end{enumerate}
Set $T_{n+1}^* = T_{n+1}$ if $B_{n+1}=1$.
Combining cases 1--3 gives that when the new MBI process $\{\bxi_{N(n)+1,\delta}(t-T^*_{n+1}):t\ge T^*_{n+1}\}$ is initiated, its initial value satisfies
\[
\EE\left(\bs^{\bxi_{N(n)+1,\delta}(0)}\middle|B_{n+1}=1\right) = \frac{\alpha(1-\rho)}{\alpha+\gamma}s_2 + \frac{\gamma(1-\rho)}{\alpha+\gamma}s_1 +\rho s_1s_2,\qquad \bs\in [0,1]^2.
\]
We set $\bxi_{N(n)+1,\delta}(0) = (0,0)$ if $B_{n+1}=0$, then
$\bxi_{N(n)+1,\delta}(0)$ is a random vector with generating function
\[
\EE\left(\bs^{\bxi_{N(n)+1,\delta}(0)}\right) = {\alpha(1-\rho)}s_2 + {\gamma(1-\rho)}s_1 
+(\alpha+\gamma)\rho s_1s_2 + \beta,\qquad \bs\in [0,1]^2.
\]

When $B_{n+1}=0$, we
let $T_{n+1}^*$ be the first time after $T_{n+1}$ when a jump of the desired type occurs among the MBI processes in \eqref{eq:MBIs}.
In addition, when $B_{n+1}=0$, let $T_{n+1}'$ be the last time before $T_{n+1}^*$ when a jump with an undesired type happens. Then for $\Delta T_{n+1}:= (1-B_{n+1})(T_{n+1}^*-T'_{n+1})$, we have
\[
\PP^{\mathcal{F}_{T_n^*}}\left(\Delta T_{n+1}>t\middle| B_{n+1}=0\right)
= \exp\left\{-t(1+\beta)(n+1+\delta N(n))\right\},
\]
Since we do not count all jumps until the desired type of jump appears, then during the time interval $(T_n^*, T_{n+1}^*]$, the effective amount of evolution time for the MBI processes in \eqref{eq:MBIs} is $T_{n+1}-T_n^*+ \Delta T_{n+1}$, and
\begin{align}
\label{eq:cond_expT}
\EE^{\mathcal{F}_{T_n^*}} (T_{n+1}-T_n^*+ \Delta T_{n+1}) = \frac{1}{n+1+\sum_{k=1}^n R_k+\delta N(n)}.
\end{align}

Write 
\[
{\bxi}^*_\delta(T^*_n):=\left(\bxi_{1,\delta}(T^*_n),\bxi_{2,\delta}(T^*_n-T^*_{S_2}),\ldots, \bxi_{N(n),\delta}(0),
(0,0),\ldots\right), \qquad n\ge 0,
\]
then the embedding framework described above shows that
$\{{\bxi}^*_\delta(T^*_n): n\ge 0\}$ is Markovian on $\left(\mathbb{N}^2\right)^\infty$.
In the following theorem, we embed the evolution of in- and out-degree processes into the MBI  framework, and the embedding results later play a key role in the derivation of asymptotic results in Theorems~\ref{thm:limitNij} and \ref{thm:MRV}.
\begin{Theorem}\label{thm:embed_MBI}
In $\left(\mathbb{N}^2\right)^\infty$, define the in- and out-degree sequences as
\[
\bD(n) := \left(\bigl(\Din_1(n),\Dout_1(n)\bigr),\ldots, \bigl(\Din_{|V(n)|}(n),\Dout_{|V(n)|}(n)\bigr),
(0,0),\ldots\right).
\]
Then for $\{T^*_k: k\ge 0\}$ and $\{\bxi_{k,\delta}(t-T^*_{S_k}): t\ge T^*_{S_k}\}_{k\ge 1}$ constructed above,
we have that in $\left(\left(\mathbb{N}^2\right)^\infty\right)^\infty$,
\begin{align*}
\bigl\{\bD(n): n\ge 0\bigr\} \stackrel{d}{=} \left\{{\bxi}^*_\delta(T^*_n): n\ge 0\right\}.
\end{align*}
\end{Theorem}
\begin{proof}



By the model description in Section~\ref{sec:model}, $\{\bD(n):n\ge 0\}$ is Markovian on $\left(\mathbb{N}^2\right)^\infty$, so
it suffices to check the transition probability from $\bD(n)$ to $\bD(n+1)$ agrees with that from 
${\bxi}^*_\delta(T^*_n)$ to ${\bxi}^*_\delta(T^*_{n+1})$.
Write
\begin{align*}
\bolde_w^{(1)} &:= \left(\bigl(0,0\bigr),\ldots,\bigl(0,0\bigr),\underbrace{\bigl(1,0\bigr)}_{\text{$w$-th entry}},\bigl(0,0\bigr),\ldots\right),\\
\bolde_w^{(2)} &:= \left(\bigl(0,0\bigr),\ldots,\bigl(0,0\bigr),\underbrace{\bigl(0,1\bigr)}_{\text{$w$-th entry}},\bigl(0,0\bigr),\ldots\right),\\
\bolde_w^{(3)} &:= \left(\bigl(0,0\bigr),\ldots,\bigl(0,0\bigr),\underbrace{\bigl(1,1\bigr)}_{\text{$w$-th entry}},\bigl(0,0\bigr),\ldots\right),
\end{align*}
and let $\mathcal{G}_n$ denote the $\sigma$-field generated by the history of the network up to $n$ steps.
Then we have
\begin{align}
\PP^{\mathcal{G}_n}\left(\bD(n+1)=\bD(n)+\bolde_w^{(1)} + \bolde_{|V(n)|+1}^{(2)}\right)
&= \frac{\alpha(1-\rho)(\Din_w(n)+\delta)}{|E(n)|+\delta |V(n)|},\label{eq:PAtrans1}\\
\PP^{\mathcal{G}_n}\left(\bD(n+1)=\bD(n)+\bolde_w^{(2)}+ \bolde_{|V(n)|+1}^{(1)}\right)
&= \frac{\gamma(1-\rho)(\Dout_w(n)+\delta)}{|E(n)|+\delta |V(n)|},\label{eq:PAtrans2}\\
\PP^{\mathcal{G}_n}\left(\bD(n+1)=\bD(n)+\bolde_w^{(3)}++ \bolde_{|V(n)|+1}^{(3)}\right)
&= \rho\frac{\alpha(\Din_w(n)+\delta)+\gamma(\Dout_w(n)+\delta)}{|E(n)|+\delta |V(n)|},
\label{eq:PAtrans3}\\
\intertext{and under $\PP^{\mathcal{G}_n}$, for $i,j\in V(n)$,}
\PP^{\mathcal{G}_n}\left(\bD(n+1)=\bD(n)+\bolde_i^{(1)} + \bolde_{j}^{(2)}\right)
&= \beta(1-\rho)\frac{\Din_i(n)+\delta}{|E(n)|+\delta |V(n)|}\frac{\Dout_j(n)+\delta}{|E(n)|+\delta |V(n)|},\label{eq:PAtrans4}\\
\PP^{\mathcal{G}_n}\left(\bD(n+1)=\bD(n)+\bolde_i^{(3)}+ \bolde^{(3)}_{j}\right)
&= \beta\rho\frac{\Din_i(n)+\delta}{|E(n)|+\delta |V(n)|}\frac{\Dout_j(n)+\delta}{|E(n)|+\delta |V(n)|},\label{eq:PAtrans5}
\end{align}
where $|E(n)|-(n+1)$ follows a binomial distribution with size $n$ and success probability $\rho$,
and $|V(n)|-1$ follows a binomial distribution with size $n$ and success probability $1-\beta$.

By \eqref{eq:prob_Bk}, we see that $N(n)$ has the same distribution as $|V(n)|$.
Also, the definition of $\{R_k:k\ge 1\}$ gives that for $r_k\in \{0,1\}$, $1\le k\le n$,
\begin{align*}
\PP\left(R_k=r_k, k=1,\ldots, n\right) &= \EE\left(\ind_{\{R_{n}=r_{n}\}}\PP^{\mathcal{F}_{T_{n-1}^*}}\left(R_k=r_k, k=1,\ldots, n-1\right)\right)\\
&=\cdots = \prod_{k=1}^{n} \PP\left(R_k=r_k\right),
\end{align*}
and $\PP(R_{n}=1) = \EE(\PP^{\mathcal{F}_{T^*_{n-1}}}(R_{n}=1)) = \rho = 1-\PP(R_{n}=0)$.
Therefore, $\{R_k: 1\le k\le n\}$ are iid Bernoulli random variables with $\PP(R_1=1)=\rho$, and
$\sum_{k=1}^n R_k$ has the same distribution as $|E(n)|-(n+1)$.
Furthermore, since for $n\ge 1$,
\begin{align*}
\PP^{\mathcal{F}_{T^*_{n-1}}}\left(B_n=1, R_n=1\right)
&= \PP\left(M_1(n)\cup \{Z^{(1)}_n=1\}\right) + \PP\left(M_2(n)\cup \{Z^{(2)}_n=1\}\right) \\
&= (1-\beta)\rho = \PP^{\mathcal{F}_{T^*_{n-1}}}\left(B_n=1\right)\PP^{\mathcal{F}_{T^*_{n-1}}}\left(R_n=1\right),
\end{align*}
then $B_n$ and $R_n$ are independent from each other under both $\PP$ and $\PP^{\mathcal{F}_{T^*_{n-1}}}$.

To check the agreement between transition probabilities, we consider the seven scenarios listed in the embedding framework. 
For case 1, conditioning on $\mathcal{F}_{T_n^*}$, we have for $1\le w\le N(n)$ that
\begin{align*}
&\PP^{\mathcal{F}_{T_n^*}}\left(\bxi^*_\delta(T^*_{n+1})=\bxi^*_\delta(T^*_{n}) + \bolde_w^{(1)} + \bolde_{N(n)+1}^{(2)}\right)\\
&= \PP^{\mathcal{F}_{T_n^*}}\left(A_1(n+1), B_{n+1}=1,\bxi_{w,\delta}(T^*_{n+1}-T^*_{S_w}) = \bxi_{w,\delta}(T^*_{n}-T^*_{S_w})+(1,0)\right)\\
&= \frac{\alpha(1+\beta)}{\alpha+\beta}(1-\rho)\frac{(\alpha+\beta)(\xi^{(1)}_{w,\delta}(T_n^*-T^*_{S_w})+\delta)}{(1+\beta)(n+1+\sum_{k=1}^n R_k+\delta N(n))}\\
&= \frac{\alpha(1-\rho)(\xi^{(1)}_{w,\delta}(T_n^*-T^*_{S_w})+\delta)}{n+1+\sum_{k=1}^n R_k+\delta N(n)},
\end{align*}
which agrees with the transition probability in \eqref{eq:PAtrans1}.
Similarly, for case 2, we have
\begin{align*}
&\PP^{\mathcal{F}_{T_n^*}}\left(\bxi^*_\delta(T^*_{n+1})=\bxi^*_\delta(T^*_{n}) + \bolde_w^{(2)} + \bolde_{N(n)+1}^{(1)}\right)\\
&= \PP^{\mathcal{F}_{T_n^*}}\left(A_2(n+1), B_{n+1}=1,\bxi_{w,\delta}(T^*_{n+1}-T^*_{S_w}) = \bxi_{w,\delta}(T^*_{n}-T^*_{S_w})+(0,1)\right)\\
&= \frac{\gamma(1-\rho)(\xi^{(2)}_{w,\delta}(T_n^*-T^*_{S_w})+\delta)}{n+1+\sum_{k=1}^n R_k+\delta N(n)},
\end{align*}
thus giving the agreement with \eqref{eq:PAtrans2}. In case 3, we consider the occurrence of either $M_1(n+1)$ or $M_2(n+1)$, which leads to 
\begin{align*}
&\PP^{\mathcal{F}_{T_n^*}}\left(\bxi^*_\delta(T^*_{n+1})=\bxi^*_\delta(T^*_{n}) + \bolde_w^{(3)} + \bolde_{N(n)+1}^{(3)}\right)\\
&= \PP^{\mathcal{F}_{T_n^*}}\left(M_1(n+1), B_{n+1}=1,\bxi_{w,\delta}(T^*_{n+1}-T^*_{S_w}) = \bxi_{w,\delta}(T^*_{n}-T^*_{S_w})+(1,1)\right)\\
&\qquad + \PP^{\mathcal{F}_{T_n^*}}\left(M_2(n+1), B_{n+1}=1,\bxi_{w,\delta}(T^*_{n+1}-T^*_{S_w}) = \bxi_{w,\delta}(T^*_{n}-T^*_{S_w})+(1,1)\right)\\
&= \rho\frac{\alpha(\xi^{(1)}_{w,\delta}(T_n^*-T^*_{S_w})+\delta)+\gamma(\xi^{(2)}_{w,\delta}(T_n^*-T^*_{S_w})+\delta)}{n+1+\sum_{k=1}^n R_k+\delta N(n)},
\end{align*}
which agrees with \eqref{eq:PAtrans3}.

In case 4, we see that for $i\neq j$, 
\begin{align*}
&\PP^{\mathcal{F}_{T_n^*}}\left(A_1(n+1), B_{n+1}=0,\bxi_{i,\delta}(T^*_{n+1}-T^*_{S_i}) = \bxi_{i,\delta}(T^*_{n}-T^*_{S_i})+(1,0),\,\right.\\
&\left.\qquad\qquad
 \bxi_{j,\delta}(T^*_{n+1}-T^*_{S_j}) = \bxi_{j,\delta}(T^*_{n}-T^*_{S_j})+(0,1) \right)\\
&= \frac{(1-\alpha)\beta}{\alpha+\beta}(1-\rho)\frac{(\alpha+\beta)(\xi^{(1)}_{i,\delta}(T_n^*-T^*_{S_i})+\delta)}{(1+\beta)(n+1+\sum_{k=1}^n R_k+\delta N(n))}\\
&\times\sum_{k=0}^\infty\left(1-(1-\rho)\frac{(\beta+\gamma)\sum_{k=1}^{N(n)}(\xi^{(2)}_{k,\delta}(T^*_n-T^*_{S_k})+\delta)}{(1+\beta)(n+1+\sum_{k=1}^n R_k+\delta N(n))}\right)^k \frac{(1-\rho)(\beta+\gamma)(\xi^{(2)}_{j,\delta}(T^*_n-T^*_{S_j})+\delta)}{(1+\beta)(n+1+\sum_{k=1}^n R_k+\delta N(n))}\\
&= \frac{(1-\alpha)\beta}{1+\beta}(1-\rho)\frac{\xi^{(1)}_{i,\delta}(T_n^*-T^*_{S_i})+\delta}{n+1+\sum_{k=1}^n R_k+\delta N(n)}\frac{\xi^{(2)}_{j,\delta}(T^*_n-T^*_{S_j})+\delta}{n+1+\sum_{k=1}^n R_k+\delta N(n)},
\end{align*}
and for $i=j$,
\begin{align*}
&\PP^{\mathcal{F}_{T_n^*}}\left(A_1(n+1), B_{n+1}=0,\bxi_{i,\delta}(T^*_{n+1}-T^*_{S_i}) = \bxi_{i,\delta}(T^*_{n}-T^*_{S_i})+(1,0)+(0,1)\right)\\
&= \frac{(1-\alpha)\beta}{1+\beta}(1-\rho)\frac{\xi^{(1)}_{i,\delta}(T_n^*-T^*_{S_i})+\delta}{n+1+\sum_{k=1}^n R_k+\delta N(n)}\frac{\xi^{(2)}_{i,\delta}(T^*_n-T^*_{S_i})+\delta}{n+1+\sum_{k=1}^n R_k+\delta N(n)}.
\end{align*}
Similarly, we have for case 5 and $i\neq j$,
\begin{align*}
&\PP^{\mathcal{F}_{T_n^*}}\left(A_2(n+1), B_{n+1}=0,\bxi_{j,\delta}(T^*_{n+1}-T^*_{S_j}) = \bxi_{j,\delta}(T^*_{n}-T^*_{S_j})+(0,1),\,\right.\\
&\left.\qquad\qquad
 \bxi_{i,\delta}(T^*_{n+1}-T^*_{S_i}) = \bxi_{i,\delta}(T^*_{n}-T^*_{S_i})+(1,0) \right)\\
&= \frac{(1-\gamma)\beta}{\beta+\gamma}(1-\rho)\frac{(\beta+\gamma)(\xi^{(2)}_{j,\delta}(T^*_{n}-T^*_{S_j})+\delta)}{(1+\beta)(n+1+\sum_{k=1}^n R_k+\delta N(n))}\\
\times&\sum_{k=0}^\infty\left(1-(1-\rho)\frac{(\alpha+\beta)\sum_{k=1}^{N(n)}(\xi^{(1)}_{k,\delta}(T^*_n-T^*_{S_k})+\delta)}{(1+\beta)(n+1+\sum_{k=1}^n R_k+\delta N(n))}\right)^k \frac{(1-\rho)(\alpha+\beta)(\xi^{(1)}_{i,\delta}(T^*_n-T^*_{S_i})+\delta)}{(1+\beta)(n+1+\sum_{k=1}^n R_k+\delta N(n))}\\
&= \frac{(1-\gamma)\beta}{1+\beta}(1-\rho)\frac{\xi^{(1)}_{i,\delta}(T_n^*-T^*_{S_i})+\delta}{n+1+\sum_{k=1}^n R_k+\delta N(n)}\frac{\xi^{(2)}_{j,\delta}(T^*_n-T^*_{S_j})+\delta}{n+1+\sum_{k=1}^n R_k+\delta N(n)},
\end{align*}
and for $i=j$,
\begin{align*}
&\PP^{\mathcal{F}_{T_n^*}}\left(A_2(n+1), B_{n+1}=0,\bxi_{i,\delta}(T^*_{n+1}-T^*_{S_i}) = \bxi_{i,\delta}(T^*_{n}-T^*_{S_i})+(1,0)+(0,1)\right)\\
&= \frac{(1-\gamma)\beta}{1+\beta}(1-\rho)\frac{\xi^{(1)}_{i,\delta}(T_n^*-T^*_{S_i})+\delta}{n+1+\sum_{k=1}^n R_k+\delta N(n)}\frac{\xi^{(2)}_{i,\delta}(T^*_n-T^*_{S_i})+\delta}{n+1+\sum_{k=1}^n R_k+\delta N(n)}.
\end{align*}
Hence, combining cases 4 and 5 gives that for $i\neq j$,
\begin{align}
&\PP^{\mathcal{F}_{T_n^*}}\left(B_{n+1}=0,\bxi_{i,\delta}(T^*_{n+1}-T^*_{S_i}) = \bxi_{i,\delta}(T^*_{n}-T^*_{S_i})+(1,0),\,\right.\nonumber\\
&\left.\qquad\qquad \bxi_{j,\delta}(T^*_{n+1}-T^*_{S_j}) = \bxi_{j,\delta}(T^*_{n}-T^*_{S_j})+(0,1) \right)\nonumber\\
& =\beta (1-\rho) \frac{\xi^{(1)}_{i,\delta}(T_n^*-T^*_{S_i})+\delta}{n+1+\sum_{k=1}^n R_k+\delta N(n)}\frac{\xi^{(2)}_{j,\delta}(T^*_n-T^*_{S_j})+\delta}{n+1+\sum_{k=1}^n R_k+\delta N(n)},
\label{eq:beta0}
\end{align}
which corresponds to creating a new edge $(j,i)$ between two existing nodes $i,j$ but not generating any reciprocal edge.
For $i=j$,
\begin{align}
&\PP^{\mathcal{F}_{T_n^*}}\left(B_{n+1}=0,\bxi_{i,\delta}(T^*_{n+1}-T^*_{S_i}) = \bxi_{i,\delta}(T^*_{n}-T^*_{S_i})+(1,0)+(0,1)\right)\nonumber\\
&= \left(\frac{(1-\alpha)\beta}{1+\beta}+\frac{(1-\gamma)\beta}{1+\beta}\right)(1-\rho)\frac{\xi^{(1)}_{i,\delta}(T_n^*-T^*_{S_i})+\delta}{n+1+\sum_{k=1}^n R_k+\delta N(n)}\frac{\xi^{(2)}_{i,\delta}(T^*_n-T^*_{S_i})+\delta}{n+1+\sum_{k=1}^n R_k+\delta N(n)}\nonumber\\
&= \beta(1-\rho)\frac{\xi^{(1)}_{i,\delta}(T_n^*-T^*_{S_i})+\delta}{n+1+\sum_{k=1}^n R_k+\delta N(n)}\frac{\xi^{(2)}_{i,\delta}(T^*_n-T^*_{S_i})+\delta}{n+1+\sum_{k=1}^n R_k+\delta N(n)},
\label{eq:beta0_self}
\end{align}
which corresponds to adding a self loop for an existing node $i$ without any reciprocal edge.
Combining \eqref{eq:beta0} with \eqref{eq:beta0_self} gives a transition probability agreeing with \eqref{eq:PAtrans4}, i.e. given $\mathcal{F}_{T_n^*}$, for $1\le i,j\le N(n)$,
\begin{align*}
&\PP^{\mathcal{F}_{T_n^*}}\left(\bxi^*_\delta(T^*_{n+1})=\bxi^*_\delta(T^*_{n}) + \bolde_i^{(1)} + \bolde_{j}^{(2)}\right)\\
&= \beta (1-\rho) \frac{\xi^{(1)}_{i,\delta}(T_n^*-T^*_{S_i})+\delta}{n+1+\sum_{k=1}^n R_k+\delta N(n)}\frac{\xi^{(2)}_{j,\delta}(T^*_n-T^*_{S_j})+\delta}{n+1+\sum_{k=1}^n R_k+\delta N(n)}.
\end{align*}

Meanwhile, for cases 6 and 7, we have that for $i\neq j$,
\begin{align*}
&\PP\left(M_1(n+1), B_{n+1}=0,\bxi_{i,\delta}(T^*_{n+1}-T^*_{S_i}) = \bxi_{i,\delta}(T^*_{n}-T^*_{S_i})+(1,1) ,\right.\\
&\left. \qquad \bxi_{j,\delta}(T^*_{n+1}-T^*_{S_j}) = \bxi_{j,\delta}(T^*_{n}-T^*_{S_j})+(1,1) \right)\\
&= \frac{(1-\alpha)\beta}{\alpha+\beta}\rho\frac{(\alpha+\beta)(\xi^{(1)}_{i,\delta}(T_n^*-T^*_{S_i})+\delta)}{(1+\beta)(n+1+\sum_{k=1}^n R_k+\delta N(n))}\\
&\quad\times\sum_{k=0}^\infty\left(1-\rho\frac{(\beta+\gamma)\sum_{k=1}^{N(n)}(\xi^{(2)}_{k,\delta}(T^*_n-T^*_{S_k})+\delta)}{(1+\beta)(n+1+\sum_{k=1}^n R_k+\delta N(n))}\right)^k \rho\frac{(\beta+\gamma)(\xi^{(2)}_{j,\delta}(T^*_n-T^*_{S_j})+\delta)}{(1+\beta)(n+1+\sum_{k=1}^n R_k+\delta N(n))}\\
&= \frac{(1-\alpha)\beta}{1+\beta}\rho\frac{\xi^{(1)}_{i,\delta}(T_n^*-T^*_{S_i})+\delta}{1+\sum_{k=1}^n R_k+\delta N(n)}\frac{\xi^{(2)}_{j,\delta}(T^*_n-T^*_{S_j})+\delta}{n+1+\sum_{k=1}^n R_k+\delta N(n)},
\end{align*} 
and 
\begin{align*}
&\PP\left(M_2(n+1), B_{n+1}=0,\bxi_{i,\delta}(T^*_{n+1}-T^*_{S_i}) = \bxi_{i,\delta}(T^*_{n}-T^*_{S_i})+(1,1) ,\right.\\
&\left. \qquad \bxi_{j,\delta}(T^*_{n+1}-T^*_{S_j}) = \bxi_{j,\delta}(T^*_{n}-T^*_{S_j})+(1,1) \right)\\
&= \frac{(1-\gamma)\beta}{\beta+\gamma}\rho\frac{(\beta+\gamma)(\xi^{(2)}_{j,\delta}(T^*_{n}-T^*_{S_j})+\delta)}{(1+\beta)(n+1+\sum_{k=1}^n R_k+\delta N(n))}\\
& \times\sum_{k=0}^\infty\left(1-(1-\rho)\frac{(\alpha+\beta)\sum_{k=1}^{N(n)}(\xi^{(1)}_{k,\delta}(T^*_n-T^*_{S_k})+\delta)}{(1+\beta)(n+1+\sum_{k=1}^n R_k+\delta N(n))}\right)^k \rho
\frac{(\alpha+\beta)(\xi^{(1)}_{i,\delta}(T^*_n-T^*_{S_i})+\delta)}{(1+\beta)(n+1+\sum_{k=1}^n R_k+\delta N(n))}\\
&= \frac{(1-\gamma)\beta}{1+\beta}\rho\frac{\xi^{(1)}_{i,\delta}(T_n^*-T^*_{S_i})+\delta}{n+1+\sum_{k=1}^n R_k+\delta N(n)}\frac{\xi^{(2)}_{j,\delta}(T^*_n-T^*_{S_j})+\delta}{n+1+\sum_{k=1}^n R_k+\delta N(n)}.
\end{align*}
When $i=j$, we have
\begin{align*}
&\PP^{\mathcal{F}_{T_n^*}}\left(M_1(n+1), B_{n+1}=0,\bxi_{i,\delta}(T^*_{n+1}-T^*_{S_i}) = \bxi_{i,\delta}(T^*_{n}-T^*_{S_i})+(1,1)+(1,1)\right)\\
&= \frac{(1-\alpha)\beta}{1+\beta}\rho\frac{\xi^{(1)}_{i,\delta}(T_n^*-T^*_{S_i})+\delta}{n+1+\sum_{k=1}^n R_k+\delta N(n)}\frac{\xi^{(2)}_{i,\delta}(T^*_n-T^*_{S_i})+\delta}{n+1+\sum_{k=1}^n R_k+\delta N(n)},
\end{align*}
and 
\begin{align*}
&\PP^{\mathcal{F}_{T_n^*}}\left(M_2(n+1), B_{n+1}=0,\bxi_{i,\delta}(T^*_{n+1}-T^*_{S_i}) = \bxi_{i,\delta}(T^*_{n}-T^*_{S_i})+(1,1)+(1,1)\right)\\
&= \frac{(1-\gamma)\beta}{1+\beta}\rho\frac{\xi^{(1)}_{i,\delta}(T_n^*-T^*_{S_i})+\delta}{n+1+\sum_{k=1}^n R_k+\delta N(n)}\frac{\xi^{(2)}_{i,\delta}(T^*_n-T^*_{S_i})+\delta}{n+1+\sum_{k=1}^n R_k+\delta N(n)}.
\end{align*}
Therefore, combining cases 6 and 7 gives for $i\neq j$,
\begin{align}
&\PP^{\mathcal{F}_{T_n^*}}\left(B_{n+1}=0,\bxi_{i,\delta}(T^*_{n+1}-T^*_{S_i}) = \bxi_{i,\delta}(T^*_{n}-T^*_{S_i})+(1,1),\,\right.\nonumber\\
&\left.\qquad\qquad \bxi_{j,\delta}(T^*_{n+1}-T^*_{S_j}) = \bxi_{j,\delta}(T^*_{n}-T^*_{S_j})+(1,1) \right)\nonumber\\
& =\beta \rho \frac{\xi^{(1)}_{i,\delta}(T_n^*-T^*_{S_i})+\delta}{n+1+\sum_{k=1}^n R_k+\delta N(n)}\frac{\xi^{(2)}_{j,\delta}(T^*_n-T^*_{S_j})+\delta}{n+1+\sum_{k=1}^n R_k+\delta N(n)},
\label{eq:beta1}
\end{align}
which corresponds to creating a new edge $(j,i)$ between two existing nodes $i,j$, together with a reciprocal edge.
Also, for $i=j$,
\begin{align}
&\PP^{\mathcal{F}_{T_n^*}}\left(B_{n+1}=0,\bxi_{i,\delta}(T^*_{n+1}-T^*_{S_i}) = \bxi_{i,\delta}(T^*_{n}-T^*_{S_i})+(1,1)+(1,1)\right)\nonumber\\
&= \left(\frac{(1-\alpha)\beta}{1+\beta}+\frac{(1-\gamma)\beta}{1+\beta}\right)\rho\frac{\xi^{(1)}_{i,\delta}(T_n^*-T^*_{S_i})+\delta}{n+1+\sum_{k=1}^n R_k+\delta N(n)}\frac{\xi^{(2)}_{i,\delta}(T^*_n-T^*_{S_i})+\delta}{n+1+\sum_{k=1}^n R_k+\delta N(n)}\nonumber\\
&= \beta\rho\frac{\xi^{(1)}_{i,\delta}(T_n^*-T^*_{S_i})+\delta}{n+1+\sum_{k=1}^n R_k+\delta N(n)}\frac{\xi^{(2)}_{i,\delta}(T^*_n-T^*_{S_i})+\delta}{n+1+\sum_{k=1}^n R_k+\delta N(n)},
\label{eq:beta1_self}
\end{align}
which corresponds to creating a self loop $(i,i)$ for an existing node $i$, together with a reciprocal edge $(i,i)$. 
Combining \eqref{eq:beta1} with \eqref{eq:beta1_self} gives a transition probability agreeing with \eqref{eq:PAtrans5}, i.e. given $\mathcal{F}_{T_n^*}$, for $1\le i,j\le N(n)$,
\begin{align*}
&\PP^{\mathcal{F}_{T_n^*}}\left(\bxi^*_\delta(T^*_{n+1})=\bxi^*_\delta(T^*_{n}) + \bolde_i^{(3)} + \bolde_{j}^{(3)}\right)\\
&= \beta \rho \frac{\xi^{(1)}_{i,\delta}(T_n^*-T^*_{S_i})+\delta}{n+1+\sum_{k=1}^n R_k+\delta N(n)}\frac{\xi^{(2)}_{j,\delta}(T^*_n-T^*_{S_j})+\delta}{n+1+\sum_{k=1}^n R_k+\delta N(n)}.
\end{align*}
This completes the proof of Theorem~\ref{thm:embed_MBI}.
\end{proof}

\section{Proof of Theorem~\ref{thm:limitNij}}\label{proof2}

Consider
\begin{align*}
\frac{1}{n}
&\sum_{v\in V(n)} \ind_{\left\{\bigl(\Din_v(n),\Dout_v(n) \bigr) = (m,l)\right\}}\\ 
&= \frac{1}{n}\sum_{v = 2}^{|V(n)|} \ind_{\left\{\bigl(\Din_v(n),\Dout_v(n)\bigr) = (m,l)\right\}}
+\frac{1}{n}\ind_{\left\{\bigl(\Din_v(n),\Dout_v(n)\bigr) = (m,l)\right\}},
\end{align*}
and we see that the second term on the right hand side goes to 0 a.s. as $n\to\infty$.
Thus, we only focus on the first term:
\begin{align}\label{eq:PA_counts_new}
\frac{1}{n}&\sum_{v = 2}^{|V(n)|} \ind_{\left\{\bigl(\Din_v(n),\Dout_v(n)\bigr) = (m,l)\right\}},\nonumber\\
\intertext{which by Theorem~\ref{thm:embed_MBI} has the same distribution as:}
\stackrel{d}{=} 
&\frac{1}{n}\sum_{v = 2}^{N(n)} \ind_{\left\{\bxi_{v,\delta}(T^*_n-T^*_{S_v}) = (m,l)\right\}}\\
= & 
\frac{1}{n}\sum_{v =1}^{n} B_v\ind_{\left\{\bxi_{N(v),\delta}(T^*_n-T^*_{v}) = (m,l)\right\}},
\end{align}
{where the coefficient $B_v$ in front of the indicator guarantees we
sum different MBI processes inside the indicator.}

Now we divide the quantity in \eqref{eq:PA_counts_new} into different parts:
\begin{align*}
\frac{1}{n}&\left[\sum_{v = 2}^{N(n)} \ind_{\left\{\bxi_{v,\delta}(T^*_n-T^*_{S_v}) = (m,l)\right\}}
- \sum_{v=1}^n B_v\ind_{\left\{\bxi_{N(v),\delta}\left(\frac{1}{1+\rho+\delta(1-\beta)}\log(n/v)\right) = (m,l)\right\}}  \right]\\
&+ \frac{1}{n}\sum_{v = 1}^{n}\left[ B_v
\ind_{\left\{\bxi_{N(v),\delta}\left(\frac{1}{1+\rho+\delta(1-\beta)}\log(n/v)\right) = (m,l)\right\}} 
\right.\\& \left. \quad 
-\, (1-\beta)\PP\left(\bxi_{N(v),\delta}\left(\frac{1}{1+\rho+\delta(1-\beta)}\log(n/v)\right)= (m,l)\right)
\right]\\
& + \left[\frac{1-\beta}{n}\sum_{v = 1}^{n} 
\PP\left(\bxi_{N(v),\delta}\left(\frac{1}{1+\rho+\delta(1-\beta)}\log(n/v)\right)= (m,l)\right)
\right.\\& \left. \quad 
-\,(1-\beta) \int_0^1 \PP\left(-\widetilde{\bxi}_{\delta}\left(\frac{1}{1+\rho+\delta(1-\beta)}\log t\right)= (m,l)\right)\dd t
\right]\\
& + (1-\beta)\int_0^1 \PP\left(\widetilde{\bxi}_{\delta}\left(-\frac{1}{1+\rho+\delta(1-\beta)}\log t\right)= (m,l)\right)\dd t\\
&=: C_1(n) + C_2(n) + C_3(n) + C_4.
\end{align*}
Here we will show that $C_1(n)\convp 0$, $C_2(n)\convas 0$, and $C_3(n)\to 0$, as $n\to\infty$.

Since we do not count all jumps in the MBI processes over the time period $\bigcup_{k:B_k=0}(T_k, T'_k]$, then we define the effective amount of evolution time of $\{\bxi_{k,\delta}(\cdot)\}$ up to time $T^*_n$ as
\[
\widetilde{T}^*_n:=\sum_{k=1}^{n}\left(T_{k}-T^*_{k-1}+(1-B_{k})(T^*_{k}-T'_{k})\right).
\]
By \eqref{eq:cond_expT}, we apply \cite[Theorem III.9.1, Page 119]{athreya:ney:1972} to obtain that
\begin{align*}
\widetilde{T}^*_n &- \sum_{k=0}^{n-1}\frac{1}{k+1+\delta N(k)}
\end{align*}
is $L_2$-bounded martingale with respect to 
$\{\mathcal{F}_{T^*_n}: n\ge 1\}$, so
converges a.s.. Then by \cite[Corollary 2.1(iii)]{athreya:ghosh:sethuraman:2008}, we have
for $\eta>0$, 
\begin{align}
\sup_{n\eta\le v\le  n}\left|\widetilde{T}^*_n - \widetilde{T}^*_v
-\frac{1}{1+\rho+\delta (1-\beta)}\log(n/k)\right| &\convas 0,
\label{eq:conv_T}
\end{align}
as $n\to\infty$.
Also, note that 
\begin{align*}
|C_1(n)|&\le \frac{1}{n}\sum_{v=1}^{n}B_v\left|\ind_{\left\{\xi^{(1)}_{N(v),\delta}(\widetilde{T}^*_n - \widetilde{T}^*_v)= m\right\}}
\ind_{\left\{\xi^{(2)}_{N(v),\delta}(\widetilde{T}^*_n - \widetilde{T}^*_v) = l\right\}}\right.\\
 & \left. \qquad \qquad - \ind_{\left\{\xi^{(1)}_{N(v),\delta}\left(\frac{1}{1+\rho+\delta(1-\beta)}\log(n/v)\right) = m\right\}} 
\ind_{\left\{\xi^{(2)}_{N(v),\delta}\left(\frac{1}{1+\rho+\delta(1-\beta)}\log(n/v)\right) = l\right\}}  \right|\\
&\le \frac{1}{n}\sum_{v=1}^{n}B_v\left|\ind_{\left\{\xi^{(1)}_{N(v),\delta}\left(\widetilde{T}^*_n - \widetilde{T}^*_v\right)= m\right\}} - \ind_{\left\{\xi^{(1)}_{N(v),\delta}\left(\frac{1}{1+\rho+\delta(1-\beta)}\log(n/v)\right) = m\right\}} \right|\\
&\qquad +\frac{1}{n}\sum_{v=1}^{n}B_v \left|\ind_{\left\{\xi^{(2)}_{N(v),\delta}\left(\widetilde{T}^*_n - \widetilde{T}^*_v\right) = l\right\}} - \ind_{\left\{\xi^{(2)}_{N(v),\delta}\left(\frac{1}{1+\rho+\delta(1-\beta)}\log(n/v)\right) = l\right\}} \right|.
\end{align*}
Therefore,
\begin{align}
\EE |C_1(n)| &\le \frac{1}{n}\EE\left(\sum_{v=1}^{n}B_v\left|\ind_{\left\{\xi^{(1)}_{N(v),\delta}(\widetilde{T}^*_n - \widetilde{T}^*_v)= m\right\}} - \ind_{\left\{\xi^{(1)}_{N(v),\delta}\left(\frac{1}{1+\rho+\delta(1-\beta)}\log(n/v)\right) = m\right\}} \right|\right)\nonumber\\
 + &\frac{1}{n}\EE\left(\sum_{v=1}^{n}B_v\left|\ind_{\left\{\xi^{(2)}_{N(v),\delta}(\widetilde{T}^*_n - \widetilde{T}^*_v) = l\right\}} - \ind_{\left\{\xi^{(2)}_{N(v),\delta}\left(\frac{1}{1+\rho+\delta(1-\beta)}\log(n/v)\right) = l\right\}} \right|\right).
 \label{eq:R1_bound}
\end{align}
Applying the a.s. convergence results in \eqref{eq:conv_T} and using the methods from \cite[Theorem 1.2, pp 489--490]{athreya:ghosh:sethuraman:2008}, we have
 \begin{align*}
 &\left|\ind_{\{\xi^{(1)}_{N(v),\delta}(\widetilde{T}^*_n-\widetilde{T}^*_{v})=m\}}
-\ind_{\{\xi^{(1)}_{N(v),\delta}\bigl(\log(n/v)/(1+\rho+\delta(1-\beta))\bigr)=l\}}\right|\\
&\le \sup_{t\in [0,-\log\eta/(1+\rho+\delta(1-\beta))]}\PP\left(\xi^{(1)}_{2,\delta}(t+\epsilon)-
 \xi^{(1)}_{2,\delta}\bigl((t-\epsilon)\wedge 0)\bigr)\ge 1\right)\\
&\quad +\PP\left(\sup_{n\eta \le v\le n}\left|\widetilde{T}^*_n-\widetilde{T}^*_{v}-\frac{1}{1+\rho+\delta(1-\beta)}\log(n/v)\right|\ge \epsilon\right)=: p_1(\epsilon,\eta).
 \end{align*}
 Similarly,
 \begin{align*}
 &\left|\ind_{\{\xi^{(2)}_{N(v),\delta}(\widetilde{T}^*_n-\widetilde{T}^*_{v})=l\}}
-\ind_{\{\xi^{(2)}_{N(v),\delta}\bigl(\log(n/v)/(1+\rho+\delta(1-\beta))\bigr)=l\}}\right|\\
&\le \sup_{t\in [0,-\log\eta/(1+\rho+\delta(1-\beta))]}\PP\left(\xi^{(2)}_{2,\delta}(t+\epsilon)-
 \xi^{(2)}_{2,\delta}\bigl((t-\epsilon)\wedge 0)\bigr)\ge 1\right)\\
&\quad +\PP\left(\sup_{n\eta \le v\le n}\left|\widetilde{T}^*_n-\widetilde{T}^*_{v}-\frac{1}{1+\rho+\delta(1-\beta)}\log(n/v)\right|\ge \epsilon\right)=: p_2(\epsilon,\eta).
 \end{align*}
 Therefore, we see from \eqref{eq:R1_bound} that
 \begin{align*}
 \EE|C_1(n)|&\le 2\cdot\frac{1}{n}\cdot n\eta+\frac{1}{n}(1-\eta)n\bigl(p_1(\epsilon,\eta)+p_2(\epsilon,\eta)\bigr),
 \end{align*}
 which implies $\lim_{n\to\infty}\EE|C_1(n)|=0$. Therefore, $C_1(n)\convp 0$.

To prove $C_2(n)\convas 0$, first consider for $v\ge 1$, 
\begin{align*}
X_v& := B_v
\ind_{\left\{\bxi_{N(v),\delta}\left(\frac{1}{1+\rho+\delta(1-\beta)}\log\left(\frac{n}{v}\right)\right) = (m,l)\right\}} \\
&\qquad - (1-\beta)\PP\left(\bxi_{N(v),\delta}\left(\frac{1}{1+\rho+\delta(1-\beta)}\log\left(\frac{n}{v}\right)\right) = (m,l)\right)\\
&= B_v\ind_{\left\{\bxi_{N(v),\delta}\left(\frac{1}{1+\rho+\delta(1-\beta)}\log\left(\frac{n}{v}\right)\right) = (m,l)\right\}} \\
&\qquad - (1-\beta)\PP\left(\widetilde{\bxi}_{\delta}\left(\frac{1}{1+\rho+\delta(1-\beta)}\log\left(\frac{n}{v}\right)\right) = (m,l)\right).
\end{align*}
Note that $E(X_v) =0$ for all $v\ge 1$. 
Define
\begin{align*}
\widetilde{X}_v&:= 
\ind_{\left\{\bxi_{N(v),\delta}\left(\frac{1}{1+\rho+\delta(1-\beta)}\log\left(\frac{n}{v}\right)\right) = (m,l)\right\}},
\end{align*}
and for $k\ge j\ge 1$, 
we have 
\begin{align*}
\EE&(\widetilde{X}_k^3 \widetilde{X}_j) =
\PP\left(B_k\widetilde{X}_k=1, B_j\widetilde{X}_j=1\right)
= \EE\left(\ind_{\{B_k=1\}}\ind_{\{B_j=1\}}
\PP^{\{B_i\}_{i=1}^k}\left(\widetilde{X}_k=1, \widetilde{X}_j=1\right)
\right),\\
\intertext{where applying the independence among $\{\bxi_{v,\delta}(\cdot)\}_{v\ge 2}$ and $\{B_v:v\ge 1\}$, gives}
&= \EE\left(\ind_{\{B_k=1\}}
\PP^{\{B_i\}_{i=1}^k}\left(\widetilde{X}_k=1\right)
\ind_{\{B_j=1\}}\PP^{\{B_i\}_{i=1}^k}\left(\widetilde{X}_j=1\right)
\right)
\\
&= (1-\beta)^2 \PP\left[\widetilde{\bxi}_{\delta}\left(\frac{\log(n/k)}{1+\rho+\delta(1-\beta)}\right) = (m,l)\right]\PP\left[\widetilde{\bxi}_{\delta}\left(\frac{\log(n/j)}{1+\rho+\delta(1-\beta)}\right) = (m,l)\right]\\
&= \EE(\widetilde{X}_k^3)\EE( \widetilde{X}_j).
\end{align*}
Therefore, for $k\neq j$, we have
$\EE(X_k^3 X_j) = 0$. 
Similarly, for $k\neq l\neq i\neq j$, we have $\EE(X_k^2 X_l X_i)=\EE(X_k X_l X_i X_j)=0$.
Then by the Markov's inequality, for any $\epsilon>0$,
\begin{align*}
\PP&\left(\left|\frac{1}{n}\sum_{k=1}^n X_k\right|\ge \epsilon\right) \le \frac{1}{n^4 \epsilon^4} \EE \left(\sum_{k=1}^nX_k\right)^4\\
&= \frac{1}{n^4 \epsilon^4} \EE\left(\sum_{k=1}^n X_k^4 + 4\sum_{k\neq l} X_k X_l^3
+ 3\sum_{k\neq l} X_k^2 X_l^2
+ 6 \sum_{k\neq l\neq i} X_k^2 X_l X_i
+ \sum_{k\neq l\neq i\neq j} X_k X_l X_i X_j
\right)\\
& = \frac{1}{n^4 \epsilon^4} \EE\left(\sum_{k=1}^n X_k^4 
+ 3\sum_{k\neq l} X_k^2 X_l^2
\right),\\
\intertext{since $|X_k|\le 1$ for $k\ge 1$, we have}
&\le \frac{1}{n^3\epsilon^4} + \frac{3}{n^2\epsilon^2}.
\end{align*}
Then by Borel-Cantelli lemma, we have
$\frac{1}{n}\sum_{k=1}^n X_k \convas 0$,
which gives $C_2(n)\convas 0$.

For $C_3(n)$, since the function $\PP[\bxi_{2,\delta}(t)=(m,l)]$ is bounded and continuous in $t$, 
 then $C_3(n)\to 0$ by the Riemann integrability of $\PP[\bxi_{2,\delta}(-\log t/(1+\rho+\delta(1-\beta)))=(m,l)]$,
 thus completing the proof of \eqref{eq:Nij}.

\section{Generalized Breiman's Theorem}\label{breiman}
We now give the generalized Breiman's theorem \cite{breiman:1965} which is useful to show the MRV of $(\mathcal{I},\mathcal{O})$ in Theorem~\ref{thm:MRV}. This result
about products has spawned many proofs and generalizations. See for
instance
\cite{resnickbook:2007,
  kulik:soulier:2020,
  fougeres:mercadier:2012,
  maulik:resnick:rootzen:2002, chen:chen:gao:2019,
  basrak:davis:mikosch:2002b}. Here we only present 
the result proved in \cite[Theorem 3]{wang:resnick:2021b}.

\begin{Theorem}\label{th:extendBrei}
Suppose $\{\bxi(t): t\geq 0\}$ is an $\RR_+^p$-valued stochastic
process 
for some $p\geq 1$.  Let $X$ be a positive random variable with
regularly varying distribution satisfying for some scaling function $b(t)$,
$$\lim_{t\to\infty} t\PP( X/b(t) >x) =x^{-c} =:\nu_c \bigl((x,\infty)\bigr), \quad x>0, c>0.$$
Further suppose
\begin{enumerate}
\item For some finite and positive random vector $\bxi_\infty$,
  $$\lim_{ t \to \infty} {\bxi(t)} =\bxi_\infty \quad (\text{almost surely});$$
  \item The random variable $X$ and the process $\bxi (\cdot)$ are independent.
  \end{enumerate}
  Then:

  (i) In $\mathbb{M}(\RR_+^p \times (\RR_+\setminus \{0\}))$,
  \begin{equation}\label{e:beforeMult}
    t\PP\Bigl[ \Bigl({\bxi(X)}, \frac{X}{b(t)}\Bigr) \in \cdot \,\Bigr]
    \longrightarrow \PP(\bxi_\infty \in \cdot \,) \times
    \nu_c (\cdot)=:\eta(\cdot).\end{equation}
  If $\bxi_\infty$ is of the form $\bxi_\infty =:L\bv$ where $L>0$
  almost surely and $\bv \in (0,\infty)^p$, then $\eta (\cdot)$
  concentrates on the subcone $\mathcal{L}\times (\RR_+\setminus
  \{0\})$ where $\mathcal{L}=\{\theta \bv : \theta>0\}$.

  (ii) If additionally, for some $c'>c$ we have the condition
  \begin{equation}\label{e:extraCond}
  \kappa:=  \sup_{t\geq 0}   \EE\left[ \Bigl(  { \|\bxi(t) \|}  \Bigr)^{c'}\right]
    <\infty,
  \end{equation}
  for some $L_p$ norm $\|\cdot\|$,  then the product 
  of components in
  \eqref{e:beforeMult}, $\bxi(X)X $, has a regularly varying distribution with
  scaling function $b(t)$ and 
    in $\mathbb{M}(\RR_+^p \setminus
\{\bzero\})$, 
  \begin{equation}\label{e:prodOK}
    t\PP\Bigl[\frac{X\bxi(X)}{b(t)} \in \cdot \,\Bigr]
    \longrightarrow \left(\PP(\bxi_\infty \in \cdot \,) \times \nu_c\right) \circ h^{-1},
  \end{equation}
where $h(\by,x)=x\by$.
\end{Theorem}
\end{appendix}

\end{document}